\documentclass[preprint,times,1p,11pt]{elsarticle}

\let\theoremstyle\undefined
\usepackage{acronym}
\usepackage[dvipsnames]{xcolor}
\usepackage{amsmath}
\usepackage{amsthm}
\usepackage{amssymb}
\usepackage{mathtools}
\usepackage{amsfonts}
\usepackage{centernot}
\usepackage{algorithm}
\usepackage{algorithmicx}
\usepackage[noend]{algpseudocode}
\usepackage{multicol}
\usepackage{subcaption}
\usepackage{caption}
\usepackage{csquotes}
\usepackage{paralist}
\usepackage{graphicx}
\usepackage{svg}
\usepackage{tikz}
\usepackage{url}
\usetikzlibrary{arrows,automata}
\usetikzlibrary{shapes}
\usetikzlibrary{positioning}
\usetikzlibrary{patterns,patterns.meta}

\usepackage[final]{changes}

\newif\ifuseboldmathops
\newif\ifuseittextabbrevs
\useboldmathopstrue   

\ifuseittextabbrevs
	
	\newcommand{\eg}{{\it e.g.}}
	\newcommand{\ie}{{\it i.e.}}

\else
	
	\newcommand{\eg}{e.g.}
	\newcommand{\ie}{i.e.}
	
\fi

\ifuseboldmathops
\else

\fi

\ifuseboldmathops
\else

\fi

\ifuseboldmathops
\else

\fi

\ifuseboldmathops

\else

\fi

\newcommand{\rank}{\mathop{\mathsf{rank}}}

\newcommand{\abs}[1]{\lvert#1\rvert}
\newcommand{\card}[1]{\left|#1\right|}

\newcommand{\dist}[1]{\mathcal{D}(#1)}

 \newcommand{\supp}{\mbox{Supp}}

\newcommand{\calO}{\mathcal{O}}

\newcommand{\calF}{F}

\renewcommand{\vec}[1]{\mathbf{#1}}

\newcommand{\win}{\mathsf{Win}}

\acrodef{mdp}[MDP]{Markov Decision Process}

\theoremstyle{definition}
\newtheorem{definition}{Definition}
\newtheorem{example}{Example}
\newtheorem{problem}{Problem}

\newtheorem{remark}{Remark}

\newtheorem{lemma}{Lemma}
\newtheorem{assumption}{Assumption}
\newtheorem{corollary}{Corollary}
\newtheorem{proposition}{Proposition} 
\newtheorem{theorem}{Theorem}
\newtheorem{myfact}{Fact}
\theoremstyle{plain}

\newcommand{\hgame}{\mathcal{H}}

\newcommand{\game}{\mathcal{G}}

\newcommand{\outcomes}{\mathsf{Outcomes}}

\newcommand{\vod}{\mathsf{VoD}}

\newcommand{\srAct}{\mathsf{SRActs}}
\newcommand{\dswin}{\mathsf{DSWin}}
\newcommand{\daswin}{\mathsf{DASWin}}

\DeclareMathOperator*{\argmax}{\arg\!\max}

\begin{document}

\begin{frontmatter}
	
	
	
	\title{Integrated Resource Allocation and Strategy Synthesis in Safety Games on Graphs with Deception\tnoteref{t1}}
    \tnotetext[t1]{Distribution A: Approved for Public Release; Distribution is Unlimited.}
   
	
	\author[a]{Abhishek N. Kulkarni\corref{corr}}\ead{a.kulkarni2@ufl.edu}
	\author[a]{Matthew S. Cohen}\ead{cohen.matthew@ufl.edu}
	\author[b]{Charles A. Kamhoua}\ead{charles.a.kamhoua.civ@army.mil}
	\author[a]{Jie Fu}\ead{fujie@ufl.edu}
	\cortext[corr]{Corresponding author.}

	\affiliation[a]{organization={University of Florida},
		city={Gainesville},
		postcode={32608}, 
		state={FL},
		country={USA}
	}
	\affiliation[b]{organization={DEVCOM Army Research Laboratory}, 
		city={Adelphi},
		postcode={20783}, 
		state={MD},
		country={USA}
	}
	
	\begin{abstract}
		Deception plays a crucial role in strategic interactions with incomplete information. Motivated by security applications, we study a class of two-player turn-based deterministic games with one-sided incomplete information, in which player 1 (P1) aims to prevent player 2 (P2) from reaching a set of target states. In addition to actions, P1 can place two kinds of deception resources: ``traps'' and ``fake targets'' to disinform P2 about the transition dynamics and payoff of the game. Traps ``hide the real'' by making trap states appear normal, while fake targets ``reveal the fiction'' by advertising non-target states as targets. We are interested in jointly synthesizing optimal decoy placement and deceptive defense strategies for P1 that exploits P2's misinformation.  We introduce a novel hypergame on graph model and two solution concepts: stealthy deceptive sure winning and stealthy deceptive almost-sure winning. These identify states from which P1 can prevent P2 from reaching the target in a finite number of steps or with probability one without allowing P2 to become aware that it is being deceived. Consequently, determining the optimal decoy placement corresponds to maximizing the size of P1's deceptive winning region. Considering the combinatorial complexity of exploring all decoy allocations, we utilize compositional synthesis concepts to show that the objective function for decoy placement is monotone, non-decreasing, and, in certain cases, sub- or super-modular. This leads to a greedy algorithm for decoy placement, achieving a $(1 - 1/e)$-approximation when the objective function is sub- or super-modular. The proposed hypergame model and solution concepts contribute to understanding the optimal deception resource allocation and deception strategies in various security applications.
	\end{abstract}
	
	
	
	\begin{keyword}
		Hypergame theory \sep Security games \sep Strategic deception \sep Mechanism design \sep Synthesis of reactive systems 
		
		
	\end{keyword}
\end{frontmatter}

	\section{Introduction}
		\label{sec:intro}


Deception plays a key role in understanding and designing AI agents in various domains such as security and defense \cite{carroll2009game,pawlick2021game,zhu2021gametheoretic,kulkarni2020synthesis}, robotics \cite{shim2013taxonomy,wagner2011acting,arkin2012moral}, economics \cite{gneezy2005deception,gan2019imitative,pawlick2018modeling}, reinforcement learning \cite{zhu2021survey,huang2019deceptive}, and large language models \cite{park2023ai,hagendorff2023deception}. Game-theoretic deception provides valuable insights into strategic decision-making in interactions involving multiple agents. 
In this paper, we consider a joint deception resource allocation and deceptive strategy synthesis problem for a class of games on graphs with incomplete information. 
Games on graphs are widely used models for sequential decision-making in computer science, control theory, and robotics, with applications to synthesis \cite{bloem2018graph,ramadge1987supervisory,kressgazit2011correct} and verification \cite{baier2008principles}. Specifically, we consider a subclass of games on graphs called reachability games that represent a sequential interaction between two players, namely, a defender (P1) and an attacker (P2).
The attacker's objective is to reach a set of target states, while that of the defender is to prevent the attacker from reaching a target state. 
Employing the solutions of zero-sum reachability games \cite{deAlfaro2007concurrent}, we can identify a set of states from which P1 has no strategy to prevent P2 from visiting a true target. 
To protect targets when the game starts from a P1's losing position, P1 can allocate deception resources to disinform the attacker and further synthesize a deceptive strategy that exploits the attacker's misinformation to prevent it from reaching the target states.
We consider two classes of deception resources that serve the functions of \emph{hiding the real} and \emph{reveal the fiction} \cite{heckman2015cyber}.
Hiding the real refers to the defender simulating a trap to function like a real state while revealing the fiction corresponds to camouflaging a state to look like a target state for the attacker. 
Given this setup, we are interested in the following problem: \emph{How to optimally allocate the decoys so that the defender can influence the attacker into taking (or not taking) certain actions that maximize the defender's deceptive winning region?} 
Informally, the deceptive winning region is the set of game states from which the defender has a deceptive strategy to 
mislead the attacker to reach a fake target or a trap, and hence prevent the attacker from reaching a true target state.

This work extends our previous study\footnote{With the exception of Theorem~\ref{thm:dswinX-is-superadditive}, majority of the results presented in this paper are either entirely novel or extend the results from \cite{kulkarni2020decoy}. We also include new experiment results demonstrating the simultaneous placement of both traps and fake targets.} on decoy placement games \cite{kulkarni2020decoy} over an attack-defend game graph, studied in the field of cyberdeception \cite{wang2018cyber}.
In this problem, the objective of the attacker (P2) is to attain control over one of the target machines present in the network, while the defender (P1) strives to prevent such an outcome. 
To accomplish this, the attacker exploits vulnerabilities and enacts multi-stage lateral movements to escalate its privileges while the defender adjusts the network configuration in response. 
In \cite{kulkarni2020decoy}, we only considered the trap allocation. 
The current paper considers both traps and fake targets. In the cybersecurity context, traps can be realized by honey-patches \cite{araujo2014patches}, in which a vulnerability is patched but remains seen as a valid attack action to the attacker. 
Fake targets can be realized by high-interaction honeypots; for example, a decoy imitating a genuine data server may divert the attacker's attention and safeguard the real server. However, implementing the defense countermeasures can be costly. This paper provides a game-theoretic solution to design optimal defense strategies and deception resource allocations in attack-defend games on graphs. Nevertheless, the results derived in this paper also provide important insights into the rational behavior of agents in a sub-class of games on graphs with one-sided incomplete information, in which one of the players mislabels a subset of states.


\paragraph*{Contributions} The contributions of this paper are manifold. 

First, we introduce a class of \emph{hypergames on graph} to model the game with one-sided incomplete information resulting from P2's unawareness about the subset of states in the reachability game allocated by P1 as either ``traps'' or ``fake targets''. In principle, the traps alter the transition dynamics of the game without affecting P2's perception, whereas the fake targets manipulate P2's perception of the goal states in the game. Thus, by deciding the location of decoys, P1 can influence P2's perception and, therefore, its behavior. The hypergames on graph model integrates the two subjective views of the same interaction into a single graphical model, facilitating the analysis of agents' rational behavior within the game with one-sided incomplete information.

Second, we introduce two solution concepts to study the rational behavior of the players in the hypergame on graph called \emph{stealthy deceptive sure winning} and \emph{stealthy deceptive almost-sure winning}. We present procedures to synthesize each player's strategies under each of these concepts given a game on graph and a decoy placement. Intuitively, a deceptive strategy leverages P2's misperception to prevent the game from reaching a target state. At the same time, the stealthiness of the strategy prevents P2 from becoming aware of its misperception until a decoy state is reached. 


Third, we analyze the effect of traps and fake targets on P2's behavior when players follow either greedy deterministic strategies or randomized strategies. With greedy deterministic strategies, we show that \emph{fake targets could be more advantageous than traps} (Theorem~\ref{thm:traps-fakes-subset-relation}). Whereas, with randomized strategies, we find that \emph{neither the fake targets nor the traps provide a greater benefit over the other} (Theorem~\ref{thm:daswin-equal-for-fakes-traps}). Using these insights, we show a reduction from the problem of synthesizing stealthy deceptive sure (resp., almost-sure) winning strategy to that of synthesizing sure (resp., almost-sure) winning strategy (Theorems~\ref{thm:hgame-swin-is-dswin} and \ref{thm:daswin-reduction}). Moreover, we observe that the benefit of using deception is greater when players use greedy deterministic strategies than when they use randomized strategies (Theorem~\ref{thm:daswin-dswin-comparison}). This is a surprising result since, for several classes of games on graphs, randomized strategies are either equally or more powerful than the deterministic ones \cite{deAlfaro2007concurrent,chatterjee2006algorithms}.



Finally, we note that the task of determining an optimal placement of decoys that maximizes the size of the stealthy deceptive sure/almost-sure winning region poses a challenging combinatorial problem. To address this challenge and develop an algorithm with practical feasibility, we establish three key properties. Firstly, we demonstrate that the placement of traps and fake targets can be treated independently, as fake targets offer at least the same advantages as traps. Secondly, drawing insights from concepts in compositional synthesis \cite{filiot2011antichains,kulkarni2018compositional}, we establish sufficient conditions under which the objective function (i.e., the size of the stealthy deceptive sure/almost-sure winning region) exhibits submodularity or supermodularity property. Leveraging these findings, we propose a greedy algorithm (Algorithm~\ref{alg:greedymax}) to incrementally place the decoys. The algorithm is $(1-1/e)$-optimal when the objective function is sub- or super-modular. This approach alleviates the need to exhaustively solve a large number of hypergames for all possible configurations of decoys.




\paragraph*{Related Work}

Security games \cite{sinha2018stackelberg,kiekintveld2009computing} represent an important category of models for investigating strategic interactions in the context of both physical infrastructure security \cite{tambe2011security,tsai2010urban,paruchuri2008playing} as well as cybersecurity \cite{durkota2015optimal,pawlick2019gametheoretic,gan2022defense}. Various works in AI such as \cite{yang2013improving} have studied the optimal allocation of defense resources, with the objective of synthesizing a patrolling or inspection strategies for defenders. Recently, there has been a growing interest to investigating the allocation of deception resources in security games where players strategically employ deception to mitigate information disadvantages often encountered in these games. In \cite{thakoor2019cyber}, the authors formulate a security game (Stackelberg game) to allocate limited deception resources in a cybernetwork to mask network configurations from the attacker. This class of deception manipulates the adversary's perception of the payoff matrix and thus causes the adversary to take (or not to take) certain actions that aid the objective of the defender. In \cite{durkota2015optimal}, the authors formulate a Markov decision process to assess the effectiveness of a fixed honeypot allocation in an attack graph, which captures multi-stage lateral movement attacks in a cybernetwork and dependencies between vulnerabilities \cite{jha2002two,ou2006scalable}. In \cite{anwar2020honeypot}, the authors analyze the honeypot allocation problem for attack graphs using normal-form games, where the defender allocates honeypots that change the payoff matrix of players. The optimal allocation strategy is determined using the minimax theorem. The attack graph is closely related to our game on graph model, which generalizes the attack graph to \emph{attack-defend game graphs} \cite{jiang2009optimal,aslanyan2016quantitative} by incorporating the defender counter-actions in active defense.

The problem of optimal allocation of deception resources has received considerable attention in the domain of cybersecurity. We refer the readers to \cite{pawlick2019gametheoretic,zhu2021survey} for comprehensive surveys on game-theoretic approaches to cyber deception. In \cite{pibil2012game,kiekintveld2015gametheoretic}, the authors propose a game-theoretic method to place honeypots in a network so as to maximize the probability that the attacker attacks a honeypot and not a real system. In their game formulation, the defender decides where to insert honeypots in a network,  and the attacker chooses one server to attack and receives different payoffs when attacking a real system (positive reward) or a honeypot (zero reward).  The game is imperfect information as the real systems and honeypots are indistinguishable for the attacker. By the solution of imperfect information games, the defender's honeypot placement strategy is solved to minimize the attacker's rewards.


There are several key distinctions between our work and the prior work. First, our work investigates a qualitative approach to decoy placement instead of a quantitative one, which often involves solving an optimization problem over a well-defined reward/cost function. In the qualitative approach, we represent P2's goal as a Boolean reachability objective and our solution approaches build upon the solution of $\omega$-regular games on graphs \cite{deAlfaro2007concurrent}.  

Secondly, in contrast to numerous studies that adopt a Bayesian games framework for analyzing security games \cite{kiekintveld2010methods,kiekintveld2010robust}, we employ a hypergame model \cite{bennett1977toward} to represent P2's misinformation about the decoy locations and its lack of awareness about its own misperception. Bayesian games \cite{harsanyi1967games} are widely recognized as the standard model of games with incomplete information in game theory, primarily because of their ability to model any kind of incomplete information. However, Sasaki and Kijima \cite{sasaki2016hierarchical} have pointed out the limitations of Bayesian games in situations where one or more players are unaware of their misperceptions. In contrast, hypergames can effectively represent games with one-sided incomplete information, wherein different players play according to their subjective view of the game. Furthermore, much of the existing research on security games utilizing hypergames predominantly focuses on normal-form or repeated games models \cite{gharesifard2012evolution,gharesifard2014stealthy,kovach2015hypergame,wang1989solution,bennett1986hypergame}. In contrast, our approach centers on hypergames defined within the context of games on graphs, which provide a framework for modeling infinite games \cite{mcnaughton1993infinite,thomas1994finite}.
 
Third, we solve a stealthy strategy for the defender, which ensures that the defender's actions will not inform the attacker that deceptive tactics are being used. This is particularly important in situations where the attacker's behavior, once it becomes aware of the deception, could be unpredictable \cite{jajodia2016cyber}.  

The paper is structured as follows. In Section~\ref{sec:prelim}, we discuss the preliminaries of the reachability game on graph model to capture the interaction between P1 and P2 in which P2 does not know the locations of decoy resources. Section~\ref{sec:problem} formalizes our problem statement. In Section~\ref{sec:main}, we present the main results of this paper. We conclude the paper with two experiments in Section~\ref{sec:experiment} and a conclusion in Section~\ref{sec:conclusion}.

	\section{Preliminaries}
		\label{sec:prelim}
		We study an adversarial interaction between player 1 (P1) and player 2 (P2) represented as a zero-sum game on graph \cite{mcnaughton1993infinite,thomas1994finite}. In this two-player zero-sum \textit{game on graph}, P1 plays to prevent P2 from visiting a set of goal states. 

\begin{definition}[Reachability Game]
	\label{def:game}
	A two-player turn-based zero-sum reachability game with P2's reachability objective is a tuple: 
	\[
		\game =  \langle S, A, T, s_0, F \rangle,
	\]
	where
	\begin{itemize}
		\item $S = S_1 \cup S_2$ is a finite set of  states partitioned into two sets $S_1$ and $S_2$. At a state in $S_1$, P1 chooses an action. At a state  in $S_2$, P2 selects an action;
		
		\item $A = A_1 \cup A_2$ is the set of actions. $A_1$ (resp., $A_2$) is the set of actions for P1 (resp., P2); 
		
		\item $T : (S_1 \times A_1)\cup (S_2 \times A_2) \rightarrow S$ is a \emph{deterministic} transition function that maps a state-action pair to a next state. 
		
		\item $s_0 \in S$ is the initial state.
		
		\item $F\subseteq S$ is a set of states describing P2's reachability objective. All states in $F$ are sink states. We refer to $F$ as P2's goal states.
	\end{itemize}
\end{definition}

A \emph{path} in $\game$ is a (finite/infinite) ordered sequence of states $\rho = s_0s_1\ldots \in S^\omega$ such that, for any $i \geq 0$, there exists an action $a \in A$ such that $s_{i + 1} = T(s_{i}, a)$. A path is winning for P2 if an integer $i \ge 0$ exists such that $s_i \in F$. Otherwise, the path $\rho$ is winning for P1. By definition, a winning path $\rho$ for P1 ensures that, for all $i \ge 0$, we have $s_i\in S \setminus F$. A reachability game with P1's reachability objective is defined analogously.


A memoryless randomized strategy of player $i$, $i = 1,2$, is a map $\pi_i: S_i \rightarrow \dist{A_i}$. A player $i$ is said to \emph{follow a strategy $\pi_i$} if, for any state $s \in S_i$, the player selects an action in the support of $\pi_i(s)$. The set of all memoryless randomized strategies of player-$i$ is denoted by $\Pi_i$. A memoryless randomized strategy is deterministic if $\pi_i(s)$ is a Dirac delta function for all states $s \in S$. Given a state $s \in S$, a pair of strategies $\pi_1 \in \Pi_1$ and $\pi_2 \in \Pi_2$ determines a set of paths $\rho = s_0 s_1 \ldots \in S^\omega$ such that $s_0 = s$ and, for every $i \geq 0$, we have $s_{i+1} = T(s_i, a)$ with $a \in \supp(\pi_1(s_i))$ whenever $s_i \in S_1$, otherwise, $a \in \supp(\pi_2(s_i))$. We denote the set of paths $\rho$ determined by $s, \pi_1$ and $\pi_2$ by $\outcomes(s, \pi_1, \pi_2)$. Clearly, when $\pi_1$ and $\pi_2$ are deterministic, the set $\outcomes(s, \pi_1, \pi_2)$ is singleton.


A memoryless, randomized (resp., deterministic) strategy $\pi_2 \in \Pi_2$ is said to be \emph{almost-sure (resp., sure) winning for P2} at a state $s \in S$ if and only if, for any P1 strategy $\pi_1 \in \Pi_1$, every path $\rho \in \outcomes(s, \pi_1, \pi_2)$ visits $F$ with probability one  (resp., in finitely many steps). A strategy $\pi_1 \in \Pi_1$ is said to be \emph{almost-sure (resp., sure) winning for P1} if and only if, for every P2's strategy $\pi_2 \in \Pi_2$, no path in $\outcomes(s, \pi_1, \pi_2)$ visits $F$. The set of states from which P1 has an almost-sure (resp., sure) winning strategy is called P1's \emph{almost-sure (resp., sure) winning region}. P2's sure and almost-sure winning region is defined analogously.
The following facts are known about reachability games with complete information \cite{gradel2002automata,mcnaughton1993infinite}.

\begin{myfact}
	\label{fact:sw-is-asw}
	The sure and almost-sure winning region of any player in reachability games with complete information is equal. 
\end{myfact}

We denote P1 and P2's sure winning regions in $\game$ by $\win_1(G, F)$ and $\win_2(G, F)$, respectively.

\begin{myfact}
	\label{fact:determinacy}
	Reachability games with complete information enjoy memoryless determinacy. That is, from any state $s \in S$, exactly one of the players has a memoryless sure/almost-sure winning strategy.
\end{myfact}

Fact~\ref{fact:determinacy} implies that memoryless strategies are sufficient to analyze reachability games \cite{gradel2002automata}. Additionally, it is a known fact in game theory that deterministic strategies are sufficient for analyzing reachability games under the sure winning criterion, while randomized strategies are necessary under the almost-sure winning criterion \cite{deAlfaro2007concurrent}.  Consequently, this paper focuses on memoryless, deterministic strategies when discussing the sure winning condition and on memoryless, randomized strategies when addressing the almost-sure winning condition.

The sure/almost-sure winning region of P2 in $\game$ can be computed by using Algorithm~\ref{alg:zielonka}. The algorithm constructs a sequence of sets, called level-sets, $Z_0, Z_1, \ldots, Z_K$ such that, from any state in $Z_{k} \setminus Z_{k-1}$, $k > 0$, P2 has a strategy to visit $Z_0 := F$ in no more than $k$ steps. For any state $s \in Z_{K}$, we define its \emph{rank} to be the minimum number of steps in which P2 can ensure a visit to $F$ regardless of P1's strategy, denoted by $\rank_\game(s)$. Thus, $\rank_\game(s) = 0$ when $s \in F$, $\rank_\game(s) = \min \{k \mid s \in Z_k\}$ when $s \in Z_K \setminus F$, and $\rank_\game(s)=\infty$ when $s \notin Z_K$. The following properties of the level-sets constructed by Algorithm~\ref{alg:zielonka} are well-known \cite{deAlfaro2007concurrent}.

\begin{proposition}
	\label{prop:attractor-property}
	The following statements are true about the level-sets $Z_0, Z_1, \ldots, Z_K$ constructed by Algorithm~\ref{alg:zielonka}.
	\begin{enumerate}
		\item $Z_0 \subseteq Z_1 \subseteq Z_2 \ldots \subseteq Z_K$.
		\item For any sets $F_1 \subseteq F_2 \subseteq S$, we have $\win_2(\game, F_1) \subseteq \win_2(\game, F_2)$.
		\item For any sets $F_1, F_2 \subseteq S$, we have $\win_2(\game, F_1\cup F_2) = \win_2(\game, \win_2(\game, F_1) \cup \win_2(\game, F_2))$.
	\end{enumerate}
\end{proposition}

Given the level-sets constructed by Algorithm~\ref{alg:zielonka}, a memoryless sure winning strategy of P2 can be constructed as follows: Given a P2 state $s \in Z_K \setminus F$, let $D_s = \{a \in A_2 \mid s' = T(s, a) \land \rank_\game(s') < \rank_\game(s)\}$ be the set of actions $a \in A_2$ for which the next state $s' = T(s, a)$ has a strictly smaller rank than $s$. Then, any deterministic strategy $\pi_2: S \rightarrow A$ such that $\pi_2(s) \in D_s$ is a memoryless sure winning strategy for P2. Due to Fact~\ref{fact:determinacy}, the winning region of P1 is $S \setminus Z_K$. A deterministic memoryless strategy $\pi_1 : S \rightarrow A_1$ is sure winning for P1 at a P1 state $s \in S_1$ if $\pi_1(s) \in \{a \in A_1 \mid s' = T(s, a) \land s' \in S \setminus Z_K\}$.

\begin{remark}
	\label{rmk:greedy-swin-strategy}
	A deterministic, memoryless sure winning strategy of P2 that selects an action leading to a lower rank at all states is referred to as a \emph{greedy} strategy. Note that not all sure winning strategies are necessarily greedy. However, in any game where the set $Z_K \setminus F$ is non-empty, there is at least one state where the sure winning strategy includes a strictly rank-reducing action, as the final states have a rank of 0. 
\end{remark}

Given the level-sets constructed by Algorithm~\ref{alg:zielonka}, a memoryless \emph{almost-sure winning strategy} of P2 can be constructed as follows \cite{deAlfaro2007concurrent}: Given a P2 state $s \in Z_K \setminus F$, let $D_s = \{a \in A_2 \mid s' = T(s, a) \land s' \in Z_K\}$ be the set of actions $a \in A_2$ for which the next state $s' = T(s, a)$ is within the set $Z_K$. Then, any strategy $\pi_2 \in \Pi_2$ such that $\supp(\pi_2(s)) = D_s$ is a memoryless almost-sure winning strategy for P2. Similarly, given any P1 state $s \in S \setminus Z_K$, any strategy $\pi_1 \in \Pi_1$ such that $\supp(\pi_1(s)) = \{a \in A_1 \mid s' = T(s, a) \land s' \in S \setminus Z_K\}$ is a memoryless almost-sure winning strategy of P2. 


\begin{algorithm}[tb]
	\begin{algorithmic}[1]
		\item[\textbf{Inputs:}] $\game = \langle S, A, T, s_0, F \rangle$: Base game
		\item[\textbf{Outputs:}] $\win_2(G, F)$: P2's sure winning region in game $\game$ with P2's reachability objective, $F$.
		\State $Z_0 \gets \calF, k \gets 0$
		\Repeat
			\State $\mathsf{Pre}_1(Z_k) \gets \{v \in V_1 \mid \forall a \in A_1: \Delta(v, a) \in Z_k \}$
			\State $\mathsf{Pre}_2(Z_k) \gets \{v \in V_2 \mid \exists a \in A_2: \Delta(v, a) \in Z_k \}$
			\State $Z_{k+1} = Z_k \cup \mathsf{Pre}_1(Z_k) \cup \mathsf{Pre}_2(Z_k)$
			\State $k \gets k + 1$
		\Until{$Z_k \neq Z_{k-1}$}
		\State \Return $Z_k$
	\end{algorithmic}
	\caption{Zielonka's recursive algorithm \cite{mcnaughton1993infinite,zielonka1998infinite} to compute sure winning region of P2 in $\game$}
	\label{alg:zielonka}
\end{algorithm}

	\section{Problem Formulation}
		\label{sec:problem}
		Recall from our motivating problem that the traps and fake targets enable two types of deceptions, namely, \emph{hiding the real} and \emph{revealing the fiction}, respectively. In principle, the traps alter P2's perception of the true structure of the game, whereas the fake targets manipulate P2's perception of the goal states in the game. Hereafter, we use the terminology of this motivating problem to discuss key ideas in this paper. 

In this paper, we study the class of interactions between P1 and P2 characterized by the following information structure.

\begin{assumption}[Information Structure]
	\label{assume:information-structure}
	P1 knows the true game, \ie, the locations and types of all decoys. P2 is unaware of the presence of decoys. P1 knows about P2's unawareness. 
\end{assumption}

In a game with incomplete information satisfying Assumption~\ref{assume:information-structure}, the players perceive their interaction differently. P1 has complete information about the location and type of the decoys and, therefore, knows the true game. 

\begin{definition}[True Game]
	\label{def:true-game}
	\label{def:p1-perceptual-game}
	Given a base game $\game = \langle S, A, T, s_0, F \rangle$, let $X$ and $Y$ be two subsets of $S \setminus F$ such that $X \cap Y = \emptyset$. The deterministic two-player turn-based game representing the \emph{true} interaction between P1 and P2 when the states in $X$ are allocated as traps and those in $Y$ are allocated as fake targets is the tuple,
	\[
	\game_{X,Y}^1 = \langle S, A, T_{X,Y}, s_0, F \rangle,
	\]
	where 
	\begin{itemize}
		\item $S$,  $A$, $s_0$ and $F $ are defined as in Definition~\ref{def:game};
		
		\item $T_{X, Y}$ is a deterministic transition function. Given any state $s \in S$ and any action $a \in A$,
		\begin{align*}
			T_{X, Y}(s, a) = 
			\begin{cases*}
				T(s, a)	& \text{if } $s \notin X \cup Y$  \\
				s & \text{otherwise}  
			\end{cases*}
		\end{align*}
	\end{itemize}
\end{definition}

Note that the states in $\game$ which are allocated as decoys are `sink' states in $\game_{X, Y}^1$. 
Hereafter, we reserve the symbols $X, Y$ to represent traps and fake targets.

On the other hand, P2 is unaware of the presence of decoys. Therefore, in its subjective view of the game, P2 does not mark the states in $X \cup Y$ as sink states; instead, it considers the states in $Y$ to be goal states.

\begin{definition}[P2's Perceptual Game]
	\label{def:p2-perceptual-game}
	Given a base game $\game = \langle S, A, T, s_0, F \rangle$, a set $X$ of traps and a set $Y$ of fake targets, P2's perceptual game is the tuple 
	\[
	\game^2_{X,Y} = \langle S, A, T, s_0, F \cup Y \rangle,
	\]
	where 
	\begin{itemize}
		\item $S$, $A$, $T$, $s_0$ have the same meanings as Definition~\ref{def:game};
		
		\item $F\cup Y$ is a set of goal states as perceived by P2. 
	\end{itemize}
\end{definition}

\begin{remark}
	\label{rmk:p2-game-is-base-game}
	When P1 places no fake targets, \ie, $Y = \emptyset$, we have $\game_{X, Y}^2 = \game$. 
\end{remark}

Given the information structure in Assumption~\ref{assume:information-structure}, we consider the following problem:
\begin{problem}
	\label{prob:decoy-placement}
	Let $\game = \langle S, A, T, s_0, F \rangle$ be a reachability game. Determine the subsets $X, Y \subseteq S \setminus F$ of traps and fake targets that maximize the number of states from which P1 has (i) a sure winning strategy, (ii) an almost-sure winning strategy, to prevent P2 from reaching $F$, taking into account P2's incomplete information and subject to the constraints that $\card{X} \leq M$, $\card{Y} \leq N$ and $X \cap Y = \emptyset$.
\end{problem}

We introduce a running example to illustrate the key insights derived in this paper. 

\begin{example}[Running Example]
	\label{ex:running-ex-1}
	Consider the game depicted in Figure~\ref{fig:running-ex}. The game consists of $12$ states, where circular states represent P1 states and square states represent P2 states. The final states $s_0$ and $s_1$ are indicated by a double boundary. In this game, P2 aims to reach either state $s_0$ or $s_1$.
	
	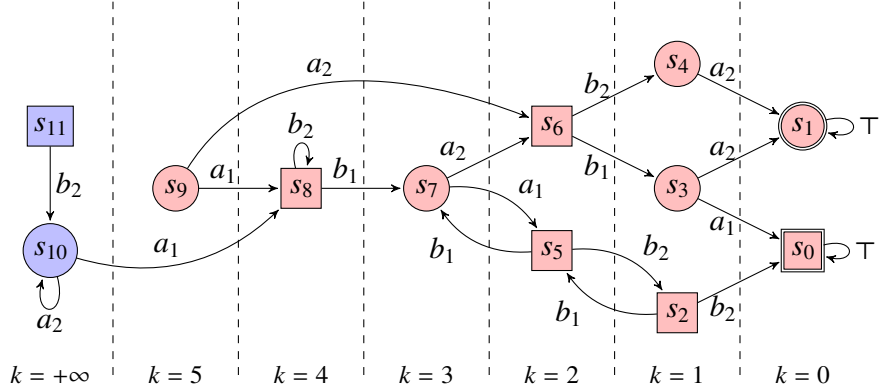
\begin{figure}[t]
		\centering 
\begin{tikzpicture}[->,>=stealth',shorten >=1pt,auto,node distance=2.5cm, scale = 0.55,transform shape]
	
	\node[state,rectangle,accepting,fill=red!25] (0) at (0, 0) {\huge $s_{0}$};
	\node[state,accepting,fill=red!25] (1) at (0, 3) {\huge $s_{1}$};
	\node[state,rectangle,fill=red!25] (2) at (-3, -1.5) {\huge $s_{2}$};
	\node[state,fill=red!25] (3) at (-3, 1.5) {\huge $s_{3}$};
	\node[state,fill=red!25] (4) at (-3, 4.5) {\huge $s_{4}$};
	\node[state,rectangle,fill=red!25] (5) at (-6, 0) {\huge $s_{5}$};
	\node[state,rectangle,fill=red!25] (6) at (-6, 3) {\huge $s_{6}$};
	\node[state,fill=red!25] (7) at (-9, 1.5) {\huge $s_{7}$};
	\node[state,rectangle,fill=red!25] (8) at (-12, 1.5) {\huge $s_{8}$};
	\node[state,fill=red!25] (9) at (-15, 1.5) {\huge $s_{9}$};
	\node[state,fill=blue!25] (10) at (-18, 0) {\huge $s_{10}$};
	\node[state,rectangle,fill=blue!25] (11) at (-18, 3) {\huge $s_{11}$};

	\path[draw=black, tips=false, dashed] (-1.5, -3) -- (-1.5,6);
	\path[draw=black, tips=false, dashed] (-4.5, -3) -- (-4.5,6);
	\path[draw=black, tips=false, dashed] (-7.5, -3) -- (-7.5,6);
	\path[draw=black, tips=false, dashed] (-10.5, -3) -- (-10.5,6);
	\path[draw=black, tips=false, dashed] (-13.5, -3) -- (-13.5,6);
	\path[draw=black, tips=false, dashed] (-16.5, -3) -- (-16.5,6);
	
	\node at (0, -3) {\LARGE $k=0$};
	\node at (-3, -3) {\LARGE $k=1$};
	\node at (-6, -3) {\LARGE $k=2$};
	\node at (-9, -3) {\LARGE $k=3$};
	\node at (-12, -3) {\LARGE $k=4$};
	\node at (-15, -3) {\LARGE $k=5$};
	\node at (-18, -3) {\LARGE $k=+\infty$};
	
	\path (0) edge[loop right]   node {\huge $\top$} (0)
	      (1) edge[loop right]   node {\huge $\top$} (1)
	      (2) edge[]   node[pos=0.3, below] {\huge $b_2$} (0)
	      (2) edge[bend left]   node[pos=0.7] {\huge $b_1$} (5)
	      (3) edge[]   node[pos=0.3, below] {\huge $a_1$} (0)
	      (3) edge[]   node[pos=0.3, above] {\huge $a_2$} (1)
	      (4) edge[]   node[pos=0.3, above] {\huge $a_2$} (1)
	      (5) edge[bend left]   node[pos=0.65] {\huge $b_2$} (2)
	      (5) edge[bend left]   node[pos=0.7] {\huge $b_1$} (7)
	      (6) edge[]   node[pos=0.3,below] {\huge $b_1$} (3)
	      (6) edge[]   node[pos=0.3,above] {\huge $b_2$} (4)
	      (7) edge[bend left]   node[pos=0.65] {\huge $a_1$} (5)
	      (7) edge[]   node[pos=0.3] {\huge $a_2$} (6)
	      (8) edge[]   node[pos=0.3] {\huge $b_1$} (7)
	      (8) edge[loop above]   node {\huge $b_2$} (8)
	      (9) edge[bend left, out=50]   node {\huge $a_2$} (6)
	      (9) edge[]   node[pos=0.3] {\huge $a_1$} (8)
	      (10) edge[bend right]   node {\huge $a_1$} (8)
	      (10) edge[loop below]   node {\huge $a_2$} (10)
	      (11) edge[]   node {\huge $b_2$} (10)
	;
\end{tikzpicture}
		\caption{Base game considered in the running example.} 
		\label{fig:running-ex} 
	\end{figure}
	
	To determine the winning regions of the players, Algorithm~\ref{alg:zielonka} is applied to $\game$ with the set of final states $F = \{s_0, s_1\}$. The colors assigned to the states in the figure correspond to the result of this algorithm. Blue-colored states represent the sure/almost-sure winning region of P1, while red-colored states represent the corresponding region of P2. Additionally, the rank of each state is indicated by its column placement. A state belonging to a column with rank $k=n$ possesses a rank equal to $n$. For instance, the states $s_5, s_6$ have a rank of $2$, while the state $s_9$ has a rank of $5$. The states $s_{10}, s_{11}$ that constitute P1's sure winning region have rank $+\infty$.

\end{example}

	\section{Main Results}
		\label{sec:main}
		Our solution to Problem~\ref{prob:decoy-placement} is based upon two fundamental problems studied in formal methods and hypergame theory, namely, deceptive planning \cite{pawlick2021game,gharesifard2012evolution,gharesifard2014stealthy,kulkarni2020decoy,kulkarni2021deceptive} and compositional synthesis \cite{baier2011compositional,filiot2011antichains,alur2016compositional,kulkarni2018compositional}. In Section~\ref{subsec:hypergames}, we leverage hypergame theory \cite{bennett1977toward} to introduce the notion of a stealthy deceptive winning strategy. This strategy enables P1 to deceive P2 into visiting a decoy state while ensuring that P2 remains unaware of the deception. Synthesizing such a stealthy deceptive strategy requires P1 to anticipate P2's strategy, given its misperception. In Section~\ref{subsec:p2-best-response}, we examine the impact of decoys on P2's perceptual game and strategy. Subsequently, in Sections~\ref{subsec:dswin} and \ref{subsec:placement}, we introduce an algorithm to synthesize a stealthy deceptive sure winning strategy and present a compositional synthesis approach to identify an approximately optimal allocation of decoys. Section~\ref{subsec:daswin} extends the synthesis algorithm to construct a stealthy deceptive almost-sure winning strategy.

	\subsection{Hypergames for Deception}
		\label{subsec:hypergames}
		A hypergame \cite{bennett1977toward} is a model used to capture strategic interactions when players have incomplete information. Intuitively, a hypergame is a game of games, and each game is associated with a player’s subjective view of its interaction with other players based on its own information and information about others’ subjective views. Hypergames are defined inductively based on the level of perception of individual players. A level-$0$ (L0) hypergame is a game with complete, symmetric information, where the perceptual games of both players are identical to the true game. In a level-$1$ (L1) hypergame, at least one of the players, say P2, misperceives the true game, but neither is aware of it. In this case, both players believe their perceptual game to be the true game and play according to their perceptual games, which are level-$0$ hypergames. In a level-$2$ (L2) hypergame, one of the players becomes aware of the misperception and is able to reason about its opponent's perceptual game. Recognizing that an L2-hypergame can represent the information structure given in Assumption~\ref{assume:information-structure}, we use an L2-hypergame\footnote{In general, it is possible to define hypergames of an arbitrary level. We refer the interested readers to \cite{wang1989solution} for an elaborate discussion on the higher levels of hypergames.} to model Problem~\ref{prob:decoy-placement}.

\begin{definition}[Level-1 and Level-2 Hypergame] 
	\label{def:generalized-hypergame}
	Given the true  game known to P1 $\game^1_{X,Y}$  and P2's perceptual game $\game^2_{X,Y}$, the level-1 (L1) hypergame is defined as a tuple $\hgame_1(X, Y) \coloneqq \langle \game^1_{X,Y}, \game^2_{X,Y} \rangle$. The level-2 (L2) hypergame between P1 and P2 when Assumption~\ref{assume:information-structure} holds is the tuple, $$\hgame_2(X, Y) = \langle \hgame_1(X, Y), \game^2_{X,Y} \rangle.$$
\end{definition}

In L2-hypergame, P1 is aware of P2's misperception, but P2 remains unaware that it lacks information. Consequently, P2 computes its strategy by solving its perceptual game $\game^2_{X, Y}$. P1 decides its strategy by solving the L1-hypergame $\hgame_1(X, Y)$, which allows P1 to incorporate P2's strategy as computed in $\game^2_{X, Y}$ into its decision-making. 

When there is a mismatch between P2's strategy at a state in its perceptual game and the same state in the true game, P1 might gain an opportunity to \emph{deceive} P2 into following a strategy that benefits P1. We call such a P1 strategy that encourages P2 to deviate from its best response in the true game $\game_{X, Y}^1$ as a deceptive strategy (formalized in Definition~\ref{def:deceptive-strategy}). However, in addition to being deceptive, P1's strategy must also be \emph{stealthy} in many practical scenarios. That is, the strategy should prevent P2 from learning that it is being deceived. For instance, if the actions chosen by P1's deceptive strategy contradict what P2 considers to be a rational strategy of P1, then P2 would become aware of its misperception. After this point, P2's behavior may become complex \cite{horak2017manipulating,wickstrom2000hawthorne,vrij2006insight,mitkidis2023morality} and modeling it is out of the scope of this paper. Hence, we focus on strategies that are both stealthy and deceptive in this paper.

For a P1's deceptive strategy to be stealthy, it must be \emph{subjectively rationalizable for P2}. Intuitively, a player's strategy is subjectively rationalizable for P2 if it is a sure/almost-sure winning strategy for that player in P2's perceptual game. Subsequently, we introduce two solution concepts: stealthy deceptive sure winning strategy and stealthy deceptive almost-sure winning strategy. 

\begin{definition}[Subjectively Rationalizable Action]
	\label{def:sracts}
	Let $\game^2_{X, Y}$ be P2's perceptual game. A P2 action $a \in A_2$ is said to be \emph{subjectively rationalizable for P2} at a state $s \in S_2 \cap \win_2(\game_{X, Y}^2, F \cup Y)$ under sure (resp., almost-sure) winning condition in $\game^2_{X, Y}$ if there exists a P2's sure (resp., almost-sure) strategy $\pi \in \Pi_2$ such that $a \in \supp(\pi(s))$. At every state $s \notin \win_2(\game_{X, Y}^2, F \cup Y)$, \ie, at all P1 states and P2 states not in $\win_2(\game_{X, Y}^2, F \cup Y)$, every enabled action at $s$ (regardless P1's or P2's action) is subjectively rationalizable for P2. 
\end{definition}


\begin{definition}[Subjectively Rationalizable Strategy]
	Given any state $s \in S$, let $\srAct(s)$ be the set of all subjectively rationalizable actions at a state $s$. A strategy $\pi \in \Pi_1 $ or $\pi\in  \Pi_2$ is said to be subjectively rationalizable for P2 if, for every $s \in S$, $\pi(s)$ is a distribution over a non-empty subset of $\srAct(s)$. 
\end{definition}

Using subjectively rationalizable strategies, we formalize the idea of a stealthy deceptive strategy for P1 under sure and almost-sure winning conditions.

\begin{definition}[Stealthy Deceptive Sure/Almost-sure Winning Strategy]
	\label{def:deceptive-strategy}
	A P1 strategy $\pi_1$ at a state $s \in S$ is said to be \emph{stealthy deceptively sure (resp., almost-sure) winning} in $\hgame_2(X, Y)$ if $\pi_1$ is subjectively rationalizable for P2 and, for any subjectively rationalizable strategy $\pi_2$ of P2, every path in $\outcomes(s, \pi_1, \pi_2)$ visits $X \cup Y$ in finitely many steps (resp., with probability one). 
\end{definition}

A state from which P1 has a stealthy deceptive sure winning strategy is called P1's \emph{deceptive sure winning state}. The set of all such states is called P1's \emph{deceptive sure winning region} and is denoted by $\dswin_1(X, Y)$. The deceptive almost-sure winning state and region are defined analogously. We denote the deceptive almost-sure winning region by $\daswin_1(X, Y)$. Given that P2 is incapable of deception under Assumption~\ref{assume:information-structure}, the notion of a deceptive winning region for P2 is undefined.

	\subsection{Effect of Decoys on P2's Subjectively Rationalizable Strategy}
		\label{subsec:p2-best-response}
		In this subsection, we discuss the effect of decoys on P2's winning region and the subjectively rationalizable strategies in its perceptual game.

As discussed in Remark~\ref{rmk:p2-game-is-base-game}, traps do not impact P2's perception. Consequently, traps do not influence the winning regions of the players in P2's perceptual game, nor do they affect P2's subjectively rationalizable strategy. Therefore, in this subsection, we focus on the case when $Y$ is a non-empty subset of $S \setminus F$ and P2's objective in $\game_{X, Y}^2$ is to reach $F \cup Y$.

We first introduce a lemma that captures the effect of making a subset of states in $\win_2(\game, F) \setminus F$ to be P2's goal states. The lemma will aid us in proving Proposition~\ref{prop:effect-on-p2}, which summarizes the effect of decoys on the size of winning regions of P1 and P2 as perceived by P2. The lemma is general and holds for any reachability game.

\begin{lemma}
	\label{lma:rank-reduction}
	Let $\game = \langle S, A, T, s_0, F \rangle$ be a game as per Definition~\ref{def:game}. Given any $Y \subseteq \win_2(\game, F) \setminus F$, let $\game_{\emptyset, Y}^2 = \langle S, A, T_{\emptyset, Y}, s_0, F \cup Y \rangle$ be a game in which a subset of P2's winning region is marked as final states in addition to $F$. Then, the rank of any state $s \in \win_2(\game, F)$ in $\game_{\emptyset, Y}^2$ is less than or equal to its rank in $\game$.
\end{lemma}
\begin{proof}
	Recall that the rank $\rank_\game(s)$ of a state $s \in \win_2(\game, F)$ in game $\game$ is the smallest number of steps in which P2 can ensure a visit to $F$, regardless of the deterministic strategy followed by P1. By definition, in game $\game$, every path in $\outcomes_\game(s, \pi_1, \pi_2)$ from any state $s \in \win_2(\game, F)$ is ensured to visit $F$ for any valid P1 strategy $\pi_1$ and any sure winning strategy $\pi_2$ of P2. Since, in game $\game_{\emptyset, Y}^2$, the presence of fake targets does not affect the transitions from any state except those in $Y$ and all states in $F \cup Y$ are sink states, two possibilities arise for any path $\rho \in \outcomes_{\game_{\emptyset, Y}^2}(s, \pi_1, \pi_2)$: either $\rho$ visits $Y$ before visiting $F$, or $\rho$ visits $F$ without visiting $Y$. In both cases, the number of steps required to visit $F \cup Y$ is at most $\rank_\game(s)$. The rank of $s$ in $\game_{\emptyset, Y}^2$ is strictly smaller than $\rank_\game(s)$ when P2 has a sure winning strategy from $s$ to visit $Y$.
\end{proof}

Since the presence of traps does not affect P2's perception, it does not affect the ranks of the states. Hence, Lemma~\ref{lma:rank-reduction} extends naturally to games containing both types of decoys.

\begin{corollary}
	\label{cor:rank-reduction}
	For any state $s \in \win_2(\game, F)$, its rank in $\game_{X, Y}^2$ is less than or equal to its rank in $\game$.
\end{corollary}

We now introduce a proposition to summarize the effect of decoys on the size of winning regions of P1 and P2 as perceived by P2.

\begin{proposition}
	\label{prop:effect-on-p2}
	The following statements about $\game_{X, Y}^2$ are true.
	\begin{enumerate}[(a)]
		\item If $Y \subseteq \win_1(\game, F)$, then $\win_1(\game_{X, Y}^2, F \cup Y) \subseteq \win_1(\game, F)$ and $\win_2(\game_{X, Y}^2, F \cup Y) \supseteq \win_2(\game, F)$.
		\item If $Y \subseteq \win_2(\game, F) \setminus F$, then $\win_1(\game_{X, Y}^2, F \cup Y) = \win_1(\game, F)$ and $\win_2(\game_{X, Y}^2, F \cup Y) = \win_2(\game, F)$.
	\end{enumerate}
\end{proposition}
\begin{proof}
	\textbf{(a).} Consider the statement $\win_2(\game_{X, Y}^2, F \cup Y) \supseteq \win_2(\game, F)$. Let $s \in \win_2(\game, F)$. By Corollary~\ref{cor:rank-reduction}, the rank of $s$ in $\game_{X, Y}^2$ is smaller than its rank in $\game$. Since any state with a finite rank in $\game_{X, Y}^2$ is a winning state for P2, $s \in \win_2(\game_{X, Y}^2, F \cup Y)$. The statement $\win_1(\game_{X, Y}^2, F \cup Y) \subseteq \win_1(\game, F)$ follows from $\win_2(\game_{X, Y}^2, F \cup Y) \supseteq \win_2(\game, F)$ using Fact~\ref{fact:determinacy}.
	
	\textbf{(b)} Consider the statement $\win_2(\game_{X, Y}^2, F \cup Y) = \win_2(\game, F)$. 

	\textbf{($\supseteq$).} Given any state $s \in \win_2(\game, F)$, by Corollary~\ref{cor:rank-reduction}, its rank in $\game_{X, Y}^2$ is finite. Thus, we have $\win_2(\game_{X, Y}^2, F \cup Y) \supseteq \win_2(\game, F)$.

	\textbf{($\subseteq$).} By way of contradiction, suppose there exists a state $s \in \win_2(\game_{X, Y}^2, F \cup Y)$ such that $s \notin \win_2(\game, F)$. This means that P2 has a greedy deterministic strategy, say $\pi_2$, to enforce a visit to $F \cup Y$ in $\game_{X, Y}^2$. But following $\pi_2$ in $\game$ does not induce a visit to $F$. Now, if $\pi_2$ induces a visit to $F$ in $\game_{X, Y}^2$, then it must also be a sure winning strategy for P2 in $\game$, as the presence of fake targets only affects the outgoing transitions from $Y$. Therefore, it must be the case that the following $\pi_2$ induces a visit to $Y$ in $\game_{X, Y}^2$. Since $Y \subseteq \win_2(\game, F)$, in game $\game$, P2 has a greedy deterministic strategy to enforce a visit to $F$ from any state in $Y$. Thus, by following $\pi_2$ until visiting $Y$ and then following any greedy sure winning strategy in game $\game$ to visit $F$ from $Y$, P2 can enforce a visit to $F$ from state $s$---a contradiction.

	The statement $\win_1(\game_{X, Y}^2, F \cup Y) = \win_1(\game, F)$ follows from $\win_2(\game_{X, Y}^2, F \cup Y) = \win_2(\game, F)$ using Fact~\ref{fact:determinacy}.
\end{proof}

Intuitively, Proposition~\ref{prop:effect-on-p2}(a) states that when fake targets are placed within P1's sure winning region in $\game$, P2 misperceives some states that are truly winning for P1 to be winning for itself. This is because P2 misperceives fake targets $Y$ as goal states.

Proposition~\ref{prop:effect-on-p2}(b) is particularly noteworthy, as it reveals that placing fake targets within P2's sure/almost-sure winning region in game $\game$ has no impact on the sure/almost-sure winning regions of the players in P2's perceptual game. This observation is intuitively supported by Corollary~\ref{cor:rank-reduction}, which states that the rank of a state in $\win_2(\game, F) \setminus F$ cannot increase when a subset of states from this set are assigned as fake targets. Additionally, there cannot exist a state outside $\win_2(\game, F)$ from which P2 can enforce a visit to $F \cup Y$ in $\game_{X, Y}^2$. Because, if such a state existed, then it should have been included in $\win_2(\game, F)$ since from all states in $Y$ in game $\game$, P2 has a strategy to enforce a visit to $F$.

However, the inclusion of fake targets in $\win_2(\game, F)$ results in modifying the set of strategies that are subjectively rationalizable for P2 when players use greedy sure winning strategies. This is due to the alteration of state ranks, which influence the set of subjectively rationalizable actions under the sure winning condition available at each state. The following example illustrates this phenomenon.

\begin{example}
	\label{ex:running-ex-2}
	Figure~\ref{fig:running-ex-fake-p2-perception} shows the perceptual games of P1 and P2 when a fake target is placed at the state $s_7$.
	Figure~\ref{fig:running-sub-a} shows P1's perceptual game, in which $s_7$ is marked as a sink state (see honeypot symbol). The sure winning region of P1 in this game contains the states $\{s_7, s_8, s_9, s_{10}, s_{11}\}$ (shown in blue), and that of P2 contains $\{s_0, s_1, s_2, s_3, s_4, s_5, s_6\}$ (shown in red). Figure~\ref{fig:running-sub-b} shows P2's perceptual game, where P2 misperceives $s_7$ as a target. Consequently, the sure winning region of P1 is $\{s_{10}, s_{11}\}$ (shown in blue) and that of P2 contains the states $\{s_0, \ldots, s_9\}$ (shown in red).

	Observe how the fake target $s_7$ affects the ranks of the states $s_0, \ldots, s_9$. When players use greedy sure winning strategies, $s_7$'s rank changes from $3$ in the base game (see Figure~\ref{fig:running-ex}) to $0$ in P2's perceptual game. Similarly, the states $s_5$ and $s_8$, from which P2 has a strategy to visit $s_7$ in one step, attain rank $1$ in P2's perceptual game.

	The changes to the ranks of the states affect P2's subjectively rationalizable strategy in its perceptual game. For instance, consider the action $b_2$ at state $s_5$. In the base game, $b_2$ is subjectively rationalizable for P2 because it is rank-reducing. However, in P2's perceptual game, the action $b_2$ is not rank-reducing. Therefore, it is not subjectively rationalizable. In fact, the action $b_1$, which was not rationalizable in the base game, becomes subjectively rationalizable for P2 in its perceptual game.

	\begin{figure*}[t]
		\centering
		\begin{subfigure}{0.45\textwidth}
			\centering
			\input{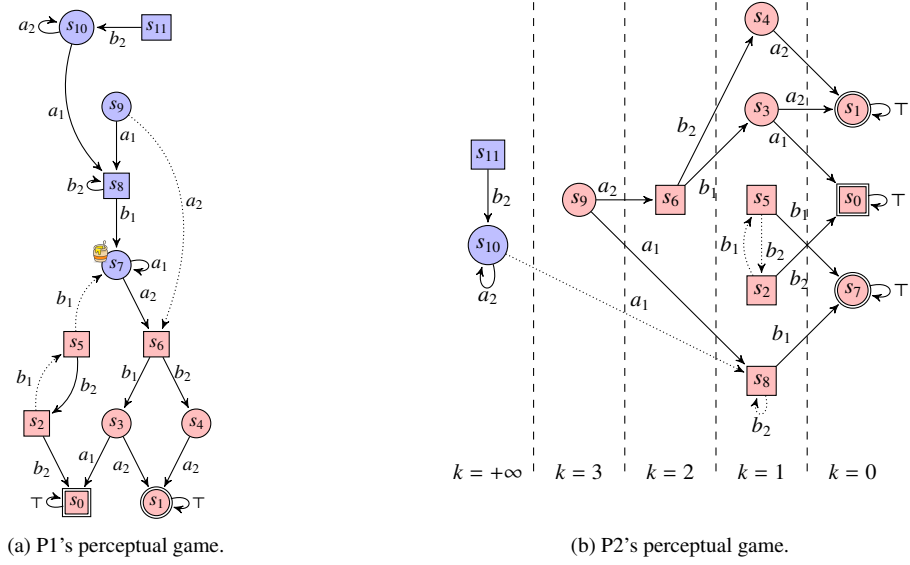}
			\caption{P1's perceptual game. }
			\label{fig:running-sub-a}
		\end{subfigure}
		\hfill
		\begin{subfigure}{0.45\textwidth}
			\centering
\begin{tikzpicture}[->,>=stealth',shorten >=1pt,auto,node distance=2.5cm, scale = 0.4,transform shape]
	
	\node[state,fill=red!25,accepting] (7) at (0, -1.5) {\huge $s_{7}$};
	\node[state,rectangle,accepting,fill=red!25] (0) at (0, 1.5) {\huge $s_{0}$};
	\node[state,accepting,fill=red!25] (1) at (0, 4.5) {\huge $s_{1}$};
	
	\node[state,rectangle,fill=red!25] (8) at (-3, -4.5) {\huge $s_{8}$};
	\node[state,rectangle,fill=red!25] (2) at (-3, -1.5) {\huge $s_{2}$};
	\node[state,rectangle,fill=red!25] (5) at (-3, 1.5) {\huge $s_{5}$};
	\node[state,fill=red!25] (3) at (-3, 4.5) {\huge $s_{3}$};
	\node[state,fill=red!25] (4) at (-3, 7.5) {\huge $s_{4}$};
	
	\node[state,fill=red!25] (9) at (-9, 1.5) {\huge $s_{9}$};
	\node[state,rectangle,fill=red!25] (6) at (-6, 1.5) {\huge $s_{6}$};
	
	\node[state,fill=blue!25] (10) at (-12, 0) {\huge $s_{10}$};
	\node[state,rectangle,fill=blue!25] (11) at (-12, 3) {\huge $s_{11}$};
	
	\path[draw=black, tips=false, dashed] (-1.5, -7.5) -- (-1.5,8.5);
	\path[draw=black, tips=false, dashed] (-4.5, -7.5) -- (-4.5,8.5);
	\path[draw=black, tips=false, dashed] (-7.5, -7.5) -- (-7.5,8.5);
	\path[draw=black, tips=false, dashed] (-10.5, -7.5) -- (-10.5,8.5);
	
	\node at (0, -7.5) {\huge $k=0$};
	\node at (-3, -7.5) {\huge $k=1$};
	\node at (-6, -7.5) {\huge $k=2$};
	\node at (-9, -7.5) {\huge $k=3$};
	\node at (-12, -7.5) {\huge $k=+\infty$};
	
	\path (0) edge[loop right]   node {\huge $\top$} (0)
	(1) edge[loop right]   node {\huge $\top$} (1)
	(2) edge[]   node[pos=0.15, below right] {\huge $b_2$} (0)
	(2) edge[bend left,densely dotted]   node[pos=0.5] {\huge $b_1$} (5)
	(3) edge[]   node[pos=0.3, left] {\huge $a_1$} (0)
	(3) edge[]   node[pos=0.3, above] {\huge $a_2$} (1)
	(4) edge[]   node[pos=0.3, left] {\huge $a_2$} (1)
	(5) edge[densely dotted]   node[pos=0.6] {\huge $b_2$} (2)
	(5) edge[]   node[pos=0.15, above right] {\huge $b_1$} (7)
	(6) edge[]   node[pos=0.15,below right] {\huge $b_1$} (3)
	(6) edge[]   node[pos=0.3, above left] {\huge $b_2$} (4)
	(7) edge[loop right]   node {\huge $\top$} (7)
	(8) edge[]   node[pos=0.3] {\huge $b_1$} (7)
	(8) edge[loop below,densely dotted]   node {\huge $b_2$} (8)
	(9) edge[]   node[pos=0.2] {\huge $a_2$} (6)
	(9) edge[]   node[pos=0.3] {\huge $a_1$} (8)
	(10) edge[densely dotted]   node {\huge $a_1$} (8)
	(10) edge[loop below]   node {\huge $a_2$} (10)
	(11) edge[]   node {\huge $b_2$} (10)
	;

\end{tikzpicture}
			\caption{P2's perceptual game.}
			\label{fig:running-sub-b}
		\end{subfigure}
		\caption{Perceptual games when the state $s_7$ is a fake target. In both sub-figures, the blue-colored states are winning for P1, and the red-colored states are winning for P2. Dotted transitions depict actions that are not subjectively rationalizable for P2 when players use greedy deterministic strategies.}
		\label{fig:running-ex-fake-p2-perception}
	\end{figure*}

\end{example}

	\subsection{Synthesis of P1's Deceptive Sure Winning Strategy}
		\label{subsec:dswin}
		In this subsection, we introduce a new \emph{hypergame on graph} model to synthesize a stealthy deceptive sure winning strategy for P1. Two key observations facilitate our definition of the hypergame on graph: 
\begin{enumerate}[(i)]
	\item \label{obs:i} When the game starts at a P1's sure/almost-sure winning state in $\game_{X, Y}^1$, P1 can prevent the game from reaching $F$ without the use of decoys.
	\item When the game starts from a P2's sure/almost-sure winning state in $\game_{X, Y}^1$, the only way for P1 to prevent the game from visiting $F$ is by forcing a visit to a decoy state.
\end{enumerate}
As a result, P1's safety objective to prevent a visit to $F$ reduces to a reachability objective to visit   $X \cup Y$.

\begin{lemma}
	\label{lma:effect-on-p1-decoys-in-win2}
	Any P1 strategy $\pi_1$ at a state  $s \in \win_2(\game, F)$ that prevents a visit to $F$ in the true game $\game_{X, Y}^1$ must ensure a visit to a state in $X \cup Y$.
\end{lemma}
\begin{proof}
	We will focus on the case where players utilize randomized strategies, given that deterministic strategies are a special case of randomized strategies.

	By way of contradiction, suppose that there exists a strategy $\pi_1$ for P1 to prevent the game $\game_{X, Y}^1$ from reaching $F$ starting from state $s$ while ensuring that no state in $X \cup Y$ is visited. In other words, the game remains indefinitely within the set $\win_2(\game, F) \setminus (F \cup X \cup Y)$. However, by definition, for every state $s \in \win_2(\game, F)$, P2 possesses a strategy $\pi_2$ that guarantees a visit to $F$ from $s$ in the original game $\game$, regardless of P1's strategy. Therefore, if P1 follows $\pi_1$ and P2 follows $\pi_2$ in the game $\game_{X, Y}^1$, the resulting path must indefinitely remain within the set $\win_2(\game, F) \setminus (F \cup X \cup Y)$ while also visiting $F$---a contradiction. Consequently, the only way for P1 to prevent the game from reaching $F$ is by visiting the set $X \cup Y$, which contains sink states.
\end{proof}

Following Observation~(\ref{obs:i}) and Lemma~\ref{lma:effect-on-p1-decoys-in-win2}, we define our hypergame on graph model as a reachability game, in which the players only follow strategies that are subjectively rationalizable for P2 and P1's objective is to reach a decoy state. When players use greedy sure winning strategies, the set of subjectively rationalizable actions at a P2 state in $\win_2(\game, F) \setminus F$ is given by 
\begin{align}
	\label{eq:sract-det}
	\srAct(s) = \{a \in A_2 \mid s' = T(s, a) \land \mathsf{rank}_{\game_{X, Y}^2}(s') < \mathsf{rank}_{\game_{X, Y}^2}(s)\}.
\end{align}

\begin{definition}[Hypergame on Graph]
	\label{def:hgame-labeling-misperception}
	Given the game $\game$, the sets of decoys $X, Y \subseteq \win_2(\game, F)$, and a function $\srAct$ that maps every state in $\game$ to a set of subjectively rationalizable actions for P2, the hypergame on graph representing the L1-hypergame $\hgame_1$ is the tuple,
	\[
	\widehat \hgame_1(X, Y)  = \langle \win_2(\game, F), A, \widehat T_{X, Y}, s_0, X \cup Y \rangle,
	\]
	where
	\begin{itemize}
		\item $\win_2(\game, F)$ is set of states.

		\item $\widehat T_{X, Y}: S \times A  \rightarrow S$ is a \emph{deterministic} transition function such that, for any state $s \in \win_2(\game, F)$, $\widehat T_{X, Y}(s, a) = T(s, a)$ if and only if $a \in \srAct(s)$. Otherwise, $\widehat{T}_X(s, a)$ is undefined.

		\item $s_0 \in \win_2(\game, F)$ is an initial state.
		
		\item $X \cup Y \subseteq \win_2(\game, F) \setminus F$ is the set of states representing P1's reachability objective.
	\end{itemize}
\end{definition}

It is noted that the set $X \cup Y$ in $\widehat{\hgame}_1(X, Y)$ defines P1's reachability objective, not P2's objective.

\begin{theorem}
	\label{thm:hgame-swin-is-dswin}
	Every sure winning strategy of P1 in $\widehat{\hgame}(X, Y)$ is a stealthy deceptive sure winning strategy for P1 in the L2-hypergame, $\hgame_2(X, Y)$.
\end{theorem}
\begin{proof}
	Every action available to P1 and P2 in $\widehat{\hgame}(X, Y)$ is greedy and subjectively rationalizable for P2 by construction. Therefore, every sure winning strategy of P1 in $\widehat{\hgame}(X, Y)$ is greedy and subjectively rationalizable for P2. By Lemma~\ref{lma:effect-on-p1-decoys-in-win2}, the strategy is stealthy deceptive sure winning for P1 in $\hgame_2(X, Y)$.
\end{proof}

\begin{example}
	In Figure~\ref{fig:running-ex-dswin}, we present the hypergame on a graph that captures the interaction between P1 and P2 as described in Example~\ref{ex:running-ex-1}. The hypergame includes states $s_0 \ldots s_9$, representing P2's sure winning region in the base game. Dotted transitions represent P2's actions that are not subjectively rationalizable for P2 in its perceptual game and are excluded from the hypergame on graph. The result of applying Algorithm~\ref{alg:zielonka} to the hypergame on graph is shown by coloring the states of the hypergame on graph. Cyan-colored states indicate that P1 has a sure winning strategy to reach $s_7$ from those states, representing P1's stealthy deceptive sure winning region. Red-colored states indicate P2's sure winning region, from which P1 has no deceptive strategy to prevent a visit to $s_0$ or $s_1$.

	For example, P1's sure winning strategy at $s_9$ is to select action $a_1$, which leads to the P2 state $s_8$. From there, the only subjectively rationalizable action for P2 is $b_1$, which leads the game to visit the fake target. It is important to note that action $a_2$ at state $s_9$ is stealthy since it is subjectively rationalizable for P2 but not deceptive sure winning for P1, as it would lead to state $s_6$ from which P1 does not possess a strategy to prevent the game from reaching either $s_0$ or $s_1$.
	
	Now, consider states $s_2$ and $s_5$. State $s_5$ is a stealthy deceptively sure winning state for P1 because the only greedy strategy available to P2 at $s_5$ selects action $b_1$, which leads to the fake target $s_7$. Note that a strategy that selects $b_2$ at $s_5$ is not greedy because $s_5$ and $s_2$ have ranks equal to $1$. Similarly, state $s_2$ is not stealthy deceptively sure winning state for P1 because the only greedy strategy at $s_2$ is to select action $b_2$ that leads to a true final state $s_0$, which P1 aims to prevent.
 
    \begin{figure}
        \centering
\begin{tikzpicture}[->,>=stealth',shorten >=1pt,auto,node distance=2.5cm, scale = 0.5,transform shape]
	
	\node[state,rectangle,accepting,fill=red!25] (0) at (0, 0) {\huge $s_{0}$};
	\node[state,accepting,fill=red!25] (1) at (0, 3) {\huge $s_{1}$};
	\node[state,rectangle,fill=red!25] (2) at (-3, -1.5) {\huge $s_{2}$};
	\node[state,fill=red!25] (3) at (-3, 1.5) {\huge $s_{3}$};
	\node[state,fill=red!25] (4) at (-3, 4.5) {\huge $s_{4}$};
	\node[state,rectangle,fill=cyan!35] (5) at (-6, 0) {\huge $s_{5}$};
	\node[state,rectangle,fill=red!25] (6) at (-6, 3) {\huge $s_{6}$};
	\node[state,fill=cyan!35,accepting] (7) at (-9, 1.5) {\huge $s_{7}$};
	\node[state,rectangle,fill=cyan!35] (8) at (-12, 1.5) {\huge $s_{8}$};
	\node[state,fill=cyan!35] (9) at (-15, 1.5) {\huge $s_{9}$};
	
	
	
	\path (0) edge[loop right]   node {\huge $\top$} (0)
	(1) edge[loop right]   node {\huge $\top$} (1)
	(2) edge[]   node[pos=0.3, below] {\huge $b_2$} (0)
	(2) edge[bend left, densely dotted]   node[pos=0.7] {\huge $b_1$} (5)
	(3) edge[]   node[pos=0.3, below] {\huge $a_1$} (0)
	(3) edge[]   node[pos=0.3, above] {\huge $a_2$} (1)
	(4) edge[]   node[pos=0.3, above] {\huge $a_2$} (1)
	(5) edge[bend left,densely dotted]   node[pos=0.65] {\huge $b_2$} (2)
	(5) edge[bend left]   node[pos=0.7] {\huge $b_1$} (7)
	(6) edge[]   node[pos=0.3,below] {\huge $b_1$} (3)
	(6) edge[]   node[pos=0.3,above] {\huge $b_2$} (4)
	(7) edge[loop right]   node {\huge $a_1$} (8)
	(8) edge[]   node[pos=0.3] {\huge $b_1$} (7)
	(8) edge[loop above,densely dotted]   node {\huge $b_2$} (8)
	(9) edge[bend left, out=45]   node {\huge $a_2$} (6)
	(9) edge[]   node[pos=0.3] {\huge $a_1$} (8)
	;
\end{tikzpicture}
        \caption{Hypergame on graph constructed based on P1 and P2's perceptual games shown in Figure~\ref{fig:running-ex-fake-p2-perception}. Dotted lines depict P2's subjectively rationalizable actions. The cyan-colored states are stealthy deceptive sure winning states for P1, whereas the red-colored states are sure winning for P2.}
        \label{fig:running-ex-dswin}
    \end{figure}
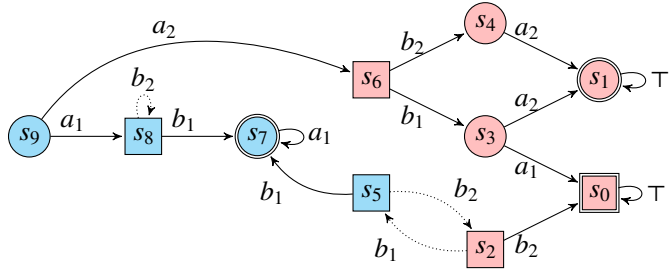
\end{example}

	\subsection{Compositional Synthesis for Decoy Placement}
		\label{subsec:placement}
		Given a placement of traps and fake targets, Theorem~\ref{thm:hgame-swin-is-dswin} provides a way to compute P1's deceptive sure winning region given a fixed decoy allocation $X, Y$. Next, we formulate a combinatorial optimization problem in which P1 aims to maximize the size of its stealthy deceptive sure winning region by allocating traps and fake targets. 

\begin{equation}
	\label{eq:problem}
	\begin{aligned}
		&X^\ast, Y^\ast = \argmax_{X, Y \subseteq \win_2(\game, F) \setminus F} \left| \dswin_1(X, Y) \right|  \\
		&\quad \mbox{subject to: } \card{X} \leq M, \card{Y} \leq N, X \cap Y = \emptyset.
	\end{aligned}
\end{equation}

In Eq.~\eqref{eq:problem}, every distinct choice of  $X, Y$ defines a hypergame, $\widehat\hgame_1(X, Y)$, which must be solved to determine the size of $\dswin_1(X, Y)$. A na\"ive approach to solving Eq.~\eqref{eq:problem} is to compute $\dswin_1(X, Y)$ for each valid placement of $X, Y$ and then select a set $X \cup Y$ for which $\abs{\dswin_1(X, Y)}$ is the largest. However, this approach is not scalable because number of hypergames to solve is ${{|\win_2(\game, F) \setminus F|}\choose{M+N}}{{M+N}\choose{M}}$, which grows rapidly with the size of game and number of decoys to place. To address this issue, we introduce a compositional approach to decoy placement in which we show that, when certain conditions hold, the decoy allocation problem can be formulated as a constrained supermodular maximization problem, for which a $(1 - \frac{1}{e})$-approximation can be computed in polynomial time using a greedy algorithm \cite{bai2018greed}.

The key insight behind our algorithm is that \emph{fake targets could be more advantageous than traps}. This enables us to decouple the placement of traps and fake targets. 

\begin{theorem}
	\label{thm:traps-fakes-subset-relation}
	For any subset $Z \subseteq \win_2(\game, F) \setminus F$, we have $\dswin_1(Z, \emptyset) \subseteq \dswin_1(\emptyset, Z)$. 
\end{theorem}
\begin{proof}
	Recall that the stealthy deceptive sure winning region in the true game is determined by computing P1's sure winning region to reach the decoys in the hypergame. Therefore, the winning regions $\dswin_1(Z, \emptyset)$ and $\dswin_1(\emptyset, Z)$ have an attractor structure. Given any $X, Y \subseteq \win_2(\game, F) \setminus F$, let $\dswin_1^i(X, Y)$ denote the $i$-th level of attractor of the sure winning region $\dswin_1(X, Y)$ in hypergame $\widehat{\hgame}_1(X, Y)$. 
	
	We will prove by induction that, for any $n \geq 0$, 
	\begin{align} \label{eq:induction-hypothesis}
		\dswin_1^n(Z, \emptyset) \subseteq \dswin_1^n(\emptyset, Z).
	\end{align}
	
	\textbf{(Base Case).} The statement is true for $n = 0$ because $\dswin_1^0(Z, \emptyset) = \dswin_1^0(\emptyset, Z) = Z$.
	
	\textbf{(Induction Step).} Let $k \geq 0$ be an integer. Suppose that Eq.~\eqref{eq:induction-hypothesis} holds for $n = k$. To show that every state $s \in \dswin_1^{k+1}(Z, \emptyset)$ is an element of $\dswin_1^{k+1}(\emptyset, Z)$, we consider two cases. 
	
	First, when $s$ is a P1 state, P1 has an action in game $\game_{Z, \emptyset}^2$ at state $s$ to visit $\dswin_1^{k}(Z, \emptyset)$ in one step. Since all P1 actions at a state in $\win_2(\game, F)$ are subjectively rationalizable for P2, due to the induction hypothesis, using the same action at $s$ would lead the game $\game_{\emptyset, Z}^2$ to visit $\dswin_1^{k}(\emptyset, Z)$ in one step. Hence, every P1 state in $\dswin_1^{k+1}(Z, \emptyset)$ is an element in $\dswin_1^{k+1}(\emptyset, Z)$. 
	
	Next, consider the case when $s$ is a P2 state. Since $s \in \dswin_1^{k+1}(Z, \emptyset)$, in game $\game_{Z, \emptyset}^2$, P1 can ensure the game to visit $Z$ in at most $(k+1)$-steps.  Now, consider the state $s$ in game $\game_{\emptyset, Z}^2$. Since $\game_{Z, \emptyset}^2 = \game$, the rank of $s$ in $\game$ (and thus $\game_{Z, \emptyset}^2 $) must be smaller than or equal to $k+1$ in game $\game_{\emptyset, Z}^2$ due to Corollary~\ref{cor:rank-reduction}. That is, $s \in \dswin_1^{k+1}(\emptyset, Z)$.  
\end{proof}

Theorem~\ref{thm:traps-fakes-subset-relation} shows that any greedy algorithm to place traps and fake targets to solve Problem~\ref{prob:decoy-placement} must place fake targets before placing the traps.

In our previous work \cite{kulkarni2020decoy}, we have studied Problem~\ref{prob:decoy-placement} when only traps are placed, \ie, $Y = \emptyset$. Hence, we first investigate how to place the fake targets to maximize the deceptive sure-winning region for P1, given only fake targets. Then, we propose an algorithm to solve Problem~\ref{prob:decoy-placement} under sure winning condition by sequentially placing the fake targets and traps. 

The concept of compositionality is important in developing a greedy algorithm for Problem~\ref{prob:decoy-placement}. It enables us to incrementally place fake targets one by one, thereby constructing $\dswin_1(\emptyset, Y)$ in an incremental manner. The following proposition states that $\dswin_1(\emptyset, Y)$ supports compositionality.

\begin{proposition}
	\label{prop:or-composition}
	Consider three placements of fake targets given by $Y_1 = \{s_1\}$, $Y_2 = \{s_2\}$, and $Y = Y_1 \cup Y_2$. Let $\dswin_1(\emptyset, Y_1)$ and $\dswin_1(\emptyset, Y_2)$ be P1's deceptive sure-winning regions in the hypergames $\widehat\hgame_1(\emptyset, Y_1)$ and $\widehat\hgame_1(\emptyset, Y_2)$, respectively. Then, P1's deceptive sure-winning region $\dswin_1(\emptyset, Y)$ in the hypergame $\widehat\hgame_1(\emptyset, Y)$ is equal to the sure-winning region for P1 in the following game:
	\begin{align*}
		\widehat\hgame_1(\emptyset, Y) = \langle \win_2(\game_{\emptyset, Y}^2, F), A, \widehat T_{\emptyset, Y},  s_0, 
		\dswin_1(\emptyset, Y_1) \cup \dswin_1(\emptyset, Y_2) \rangle,
	\end{align*}
	
	where P1's goal is to reach the target set $\dswin_1(\emptyset, Y_1) \cup \dswin_1(\emptyset, Y_2)$ and P2's goal is to prevent P1 from reaching the target set.
\end{proposition}
 
\begin{proof}
	It is observed that the underlying graphs of the three deceptive reachability games, namely $\widehat\hgame_1(\emptyset, Y_1)$, $\widehat\hgame_1(\emptyset, Y_2)$, and $\widehat\hgame_1(\emptyset, Y)$, are identical. They only differ in terms of the reachability objectives of P1. Applying Proposition~\ref{prop:attractor-property}, we have $$\dswin_1(\emptyset, Y) = \dswin_1(\emptyset, \dswin_1(\emptyset, Y_1) \cup \dswin_1(\emptyset, Y_2)),$$ which concludes the proof. 
\end{proof}

\begin{corollary}
	\label{cor:subset}
	Given a set of fake targets, $Y \subseteq \win_2(\game, F) \setminus F$ and a state $s \in \win_2(\game, F) \setminus F$, we have 
	\[
	\dswin_1(\emptyset, {Y}) \cup \dswin_1(\emptyset, {\{s\}}) \subseteq \dswin_1(\emptyset, {Y \cup \{s\}})
	\]
\end{corollary}
 
\begin{proof}
	Follows immediately by Proposition~\ref{prop:attractor-property} and the property of the sure-winning region that the goal states of a reachability objective are a subset of the sure-winning region.
\end{proof}

Thus, if we consider the size of $\dswin_1(\emptyset, {Y})$ to be a measure of the effectiveness of allocating the states in $\win_2(\game, F)$ as fake targets, then Corollary~\ref{cor:subset} states that the effectiveness of adding a new state to a set of decoys is greater than or equal to the sum of their individual effectiveness. In other words, $\dswin_1$ operator is \emph{superadditive} \cite{rosenbaum1950subadditive,hille1996functional}.

Let $\uplus$ represent the operation of composing two deceptive sure winning regions of P1. That is, given any subset $Y \subseteq \win_2(\game, F) \setminus F$ and a state $s \in \win_2(\game, F) \setminus F$,  
\[
\dswin_1(\emptyset, Y \cup \{s\}) =	\dswin_1(\emptyset, Y) \uplus \dswin_1(\emptyset, \{s\}).
\]

With this notation, the problem of optimally placing the fake targets becomes equivalent to identifying a set $Y^\ast \subseteq \win_2(\game, F) \setminus F$ such that, 
\begin{align}
	\label{eq:composition_opt}
	&Y^\ast = \argmax_{Y \subseteq \win_2(\game, F) \setminus F} \left| \biguplus\limits_{s \in Y} \dswin_1(\emptyset, \{s\}) \right|  \\
	&\quad \mbox{subject to: } \card{Y} \leq N \nonumber.
\end{align}

Let $g(Y) = \left| \biguplus\limits_{s \in Y} \dswin_1(\emptyset, \{s\}) \right|$ be a function that counts the number of P1's deceptive sure winning states when the set $Y \subseteq \win_2(\game, F) \setminus F$ is allocated as fake targets.

\begin{theorem} 
	\label{thm:dswinY-is-superadditive}
	The following statements are true.
	\begin{enumerate}[(a)]
		\item $g$ is a monotone, non-decreasing, and superadditive function. 
		
		\item $g$ is submodular if, for all  $Y \subseteq S \setminus F$ and any $s \in S \setminus F$, we have $\dswin_1(\emptyset, Y) \cup \dswin_1(\emptyset, \{s\}) = \dswin_1(\emptyset, Y \cup \{s\})$. 
		
		\item $g$ is supermodular if, for all  $Y \subseteq S \setminus F$ and any $s_1, s_2 \in S \setminus F$ and $s_1 \neq s_2$, we have $\dswin_1(X, Y \cup \{s_1\}) \cap \dswin_1(X, Y \cup \{s_2\}) = \dswin_1(X, Y)$ 
	\end{enumerate}
\end{theorem}
\begin{proof}
	(\textbf{a}). Since for any set $Y \subseteq \win_2(\game, F) \setminus F$ and any state $s \in \win_2(\game, F) \setminus (F \cup Y)$, we have $\dswin_1(\emptyset, {Y}) \cup \dswin_1(\emptyset, \{s\}) \subseteq \dswin_1(\emptyset, Y \cup \{s\})$, $\dswin_1$ is a non-decreasing, monotone function. Consequently, $g$ is also a non-decreasing monotone. The function $g$ is superadditive because, by Corollary~\ref{cor:subset}, $\dswin_1(\emptyset, {Y}) \cup \dswin_1(\emptyset, \{s\}) \subseteq \dswin_1(\emptyset, Y \cup \{s\})$. Therefore, $g(Y) + g(\{s\}) \leq g(Y \cup \{s\})$.
	
	(\textbf{b}). When $\dswin_1(\emptyset, {Y \cup \{s\}}) = \dswin_1(\emptyset, {Y}) \cup \dswin_1(\emptyset, \{s\})$, we have $g(Y) = \left| \biguplus \limits_{s \in D} \dswin_{\{s\}} \right| = \left| \bigcup \limits_{s \in D} \dswin_{\{s\}} \right|$, which is submodular \cite{vazirani2001approximation}. 
	
	(\textbf{c}). The function $g$ is supermodular if and only if
	\begin{align*}
		g(Y \cup \{s_1\}) + g(Y \cup \{s_2\}) - g(Y) \leq g(Y \cup \{s_1, s_2\}).
	\end{align*}
	Given that $\dswin_1(\emptyset, Y \cup \{s_1\}) \cap \dswin_1(\emptyset, Y \cup \{s_2\}) = \dswin_1(\emptyset, Y)$ holds for any holds for any $Y \subseteq \win_2(\game, F)$ and any $s_1, s_2 \in \win_2(\game, F)$, the LHS counts every state in $\dswin_1(\emptyset, Y \cup \{s_1\}) \cup \dswin_1(\emptyset, Y \cup \{s_2\})$ exactly once. On the other hand, RHS counts the number of states in $\dswin_1(\emptyset, Y \cup \{s_1, s_2\})$. By Proposition~\ref{prop:or-composition}, we know that RHS may contain states that are neither in $\dswin_1(\emptyset, Y \cup \{s_1\})$ nor $\dswin_1(\emptyset, Y \cup \{s_2\})$. 
\end{proof}

Given the properties of $g(Y)$, we now consider the incremental placement of traps. The following proposition, which follows from Proposition~\ref{prop:attractor-property}, provides insight into the construction of the stealthy deceptive sure winning region when traps are placed given a fixed placement of fake targets. 

\begin{proposition}
	\label{prop:or-composition-x}
	Let $\dswin_1(\{s_1\}, Y)$ and $\dswin_1(\{s_2\}, Y)$ be P1's deceptive sure-winning regions in the hypergames $\widehat\hgame_1(\{s_1\}, Y)$ and $\widehat\hgame_1(\{s_2\}, Y)$, respectively. Then, P1's deceptive sure-winning region $\dswin_1(\{s_1, s_2\}, Y)$ in the reachability game $\widehat\hgame_1(\{s_1, s_2\}, Y)$ is equal to the sure-winning region for P1 in the following game:
	\begin{align*}
		\widehat\hgame_1(\{s_1, s_2\}, Y) = \langle \win_2(\game_{X, Y}^2, F), A, \widehat T, s_0,  \dswin_1(\{s_1\}, Y) \cup \dswin_1(\{s_2\}, Y) \rangle,
	\end{align*} 
	where P1's goal is to reach the target set $\dswin_1(\{s_1\}, Y) \cup \dswin_1(\{s_2\}, Y)$
	and P2's goal is to prevent P1 from reaching the target set.
\end{proposition}

The following theorem regarding the exclusive placement of traps is from \cite{kulkarni2020decoy}. 

\begin{theorem}\cite{kulkarni2020decoy}
	\label{thm:dswinX-is-superadditive}
	For any $X \subseteq \win_2(\game, F)$, let $f(X) \mapsto \mathbb{N}$ be a function that counts the size of $\dswin_1(X, \emptyset)$. The following statements are true. 
	\begin{enumerate}[(a)]
		\item $f$ is a monotone, non-decreasing, and superadditive function. 
		
		\item $f$ is submodular if, for all  $X \subseteq S \setminus F$ and any $s \in S \setminus F$, we have $\dswin_1(X, \emptyset) \cup \dswin_1(\{s\}, \emptyset) = \dswin_1(X \cup \{s\}, \emptyset)$. 
		
		\item $f$ is supermodular if, for all  $X \subseteq S \setminus F$ and any $s_1, s_2 \in S \setminus F$ and $s_1 \neq s_2$, we have $\dswin_1(X \cup \{s_1\}, \emptyset) \cap \dswin_1(X \cup \{s_2\}, \emptyset) = \dswin_1(X, \emptyset)$ 
	\end{enumerate}
\end{theorem}

Given Theorems~\ref{thm:traps-fakes-subset-relation}, \ref{thm:dswinY-is-superadditive} and \ref{thm:dswinX-is-superadditive}, the optimal placement of decoys reduces to that of sequentially solving two superadditive function maximization problems, first maximize $g(Y)$ and then maximize $f(Y)$. However, to the best of our knowledge, there are no approximation algorithms available for maximizing superadditive functions that are applicable to to our setting. Therefore, we present Algorithm~\ref{alg:greedymax} that returns an $(1 - 1/e)$-approximate solution to Problem~\ref{prob:decoy-placement} when either condition (b) or (c) in Theorems~\ref{thm:dswinY-is-superadditive} and \ref{thm:dswinX-is-superadditive} are satisfied. This greedy algorithm is based on the GreedyMax algorithm for maximizing monotone submodular-supermodular functions in \cite{bai2018greed} and extends the algorithm discussed in \cite[Algorithm~1]{kulkarni2020decoy}.

Algorithm~\ref{alg:greedymax} starts with empty sets of states $X$ and $Y$. It first  constructs the set $Y$ by adding a new fake target in each iteration. In every step, a new fake target $s$ is selected from the set of potential decoys $D$ such that its inclusion, along with the previously chosen fake targets, maximizes the coverage of P1's deceptive sure-winning region over the states in $\win_2(\game, F)$. The process continues until either a total of $N$ fake targets have been selected, or the set of potential decoys is empty. Subsequently, the algorithm proceeds to construct $X$ using a similar procedure, where the set of fake targets $Y$ remains fixed, and a new trap is added to $X$ in each iteration.

\begin{algorithm}[tb]
	\begin{algorithmic}[1]
		\item[\textbf{Inputs:}] $\langle S, A, T, s_0, F \rangle$: Base game, $M$: Number of traps to placed, $N$: Number of fake targets to be placed.
		\item[\textbf{Outputs:}] $X, Y$:  Greedy placement of traps and fake targets. 
		\State $X\gets \emptyset$, $Y\gets \emptyset$
		\While {$N - |Y| > 0$}
		\State $D \gets \{s \in \win_2(\game, F) \mid s \notin (F \cup Y)\}$
		\If{$D$ is empty}
		\State Exit While
		\EndIf 
	\State $d \gets \arg\max_{s}\abs{\dswin_1(\emptyset, Y\cup \{s\})}$
		\State $Y \gets Y\cup \{d\}$
		\EndWhile
		\While {$M - |X| > 0$}
		\State $D \gets \{s \in \win_2(\game, F) \mid s \notin (F \cup X \cup Y)\}$
		\If{$D$ is empty}
		\State Exit While
		\EndIf 
			\State $d \gets \arg\max_{s}\abs{\dswin_1(X\cup \{s\},Y)}$
			\State $X\gets X\cup \{d\}$
		\EndWhile
		\State \Return $X, Y$
	\end{algorithmic}
	\caption{Greedy Algorithm for Decoy Placement}
	\label{alg:greedymax}
\end{algorithm}

\textbf{Complexity.} Let $V, E$ denote the number of states and transitions in the underlying graph of the hypergame $\hgame_1(X, Y)$. Then, the time complexity of Algorithm~\ref{alg:greedymax} is $\calO((V + E)\cdot(M+N)^2)$. This is because the $\dswin_1$ computation, which uses Algorithm~\ref{alg:zielonka}, has a complexity of $\calO(V + E)$ \cite{zielonka1998infinite}, and Algorithm~\ref{alg:greedymax} must solve $|\win_2(\game, F)| - |F| - j$ hypergames to determine the $j$-th decoy.

	\subsection{Synthesis of P1's Deceptive Almost-sure Winning Strategy}
		\label{subsec:daswin}
		In this section, we examine Problem~\ref{prob:decoy-placement} under the almost-sure winning criterion when players employ randomized strategies. Unlike the result from Section~\ref{subsec:placement}, we find that there is no clear advantage of either fakes or traps over the other. This difference stems from the fact that, when using randomized strategies, the players are not required to use rank-reducing strategies. The set of actions subjectively rationalizable for P2 in this case is given by 
\begin{align}
	\label{eq:sract-random}
	\widehat\srAct(s) = \{a \in A_2 \mid T(s, a) \in \win_2(\game, F)\},
\end{align}
for any state $s \in S_2 \cap \win_2(\game, F) \setminus F$. By definition, all available actions are subjectively rationalizable for P2 at every other state.

Intuitively, starting from a P2's almost-sure winning state in $\game_{X, Y}^2$, every P2 action that ensures that the game remains within the same region is subjectively rationalizable for P2. This is because (a) from every state in this region, P2 can enforce a visit to $F \cup Y$ with a positive probability, and (b) P1 has no strategy to exit this region. Therefore, a randomized strategy that selects every subjectively rationalizable action at a state with a positive probability is guaranteed to enforce a visit to $F \cup Y$ with probability one \cite{kulkarni2021deceptive}.  Such a randomized is an almost-sure winning strategy for P2 in game $\game$ \cite[Chapter~10]{baier2008principles}. 

\begin{lemma}
	\label{lma:sract-aswin}
	Let $s \in \win_2(\game, F) \setminus (F \cup Y)$ be a state in P2's perceptual game $\game_{X, Y}^2$ with the decoys $X, Y \subseteq \win_2(\game, F) \setminus F$. Then, the set of subjectively rationalizable actions at $s$ in game $\game$ is equal to that in game $\game_{X, Y}^2$.
\end{lemma}
\begin{proof}
The lemma follows from two observations. First, by Proposition~\ref{prop:effect-on-p2}, since $Y \subseteq \win_2(\game, F)$, we have $\win_2(\game, F) = \win_2(\game_{X, Y}^2, F \cup Y)$. That is, P2's winning regions in $\game$ and $\game_{X, Y}^2$ are equal. Second, by Definition~\ref{def:p2-perceptual-game}, the transitions from any state $s \notin Y$ are identical in the two games, $\game$ and $\game_{X, Y}^2$. The statement follows by the definition of $\widehat\srAct$. 
\end{proof}

Based on Lemma~\ref{lma:sract-aswin} and the knowledge that P2 follows a randomized almost-sure winning strategy in $\game_{X, Y}^2$, P1 can construct a \ac{mdp} to represent the L1-hypergame $\hgame_1$ by marginalizing the true game $\game_{X, Y}^1$ with P2's randomized almost-sure winning strategy. Since P1 does not know P2's choice of strategy, P1 would assume, in the worst case, that P2's randomized strategy may choose any subjectively rationalizable action at a given state with positive probability. This results in the following hypergame \ac{mdp} (adapted from \cite{kulkarni2021deceptive}) constructed with the following assumption.

\begin{assumption}
\label{assume:p2sra}
    At any state $s \in S_2 \cap \win_2(\game, F)$, P2 selects every subjectively rationalizable action in $\widehat{\srAct(s)}$ with a positive probability.
\end{assumption}

%
\begin{definition}[Hypergame MDP]
	\label{def:hgame-mdp}
	Given the true game $\game_{X, Y}^1$ and the function $\widehat \srAct$ that maps every state $s \in \win_2(\game, F)$ to the set of subjectively rationalizable actions for P2 at $s$, the hypergame MDP that represents L1-hypergame $\hgame_1(X, Y)$ is the following tuple,
	\[
		\widetilde \hgame_1(X, Y)  = \langle \widetilde S, A, \widetilde T_{X, Y}, X \cup Y \rangle,
	\]
	where
	\begin{itemize}
		\item $\widetilde S = \win_2(\game, F)$ is P2's sure winning region in $\game$. At P1 states in $\widetilde{S}_1 = \win_2(\game, F) \cap S_1$, P1 chooses the next action strategically. Whereas, the states in $\widetilde{S}_2 =\win_2(\game, F) \cap S_2$ are \emph{nature} states. At a nature state, the next state is chosen at random according to a predefined probability distribution.

		\item $\widehat T_{X, Y}: \widetilde S \times A \rightarrow \dist{\widetilde S}$ is a transition function defined as follows: any state $s \in X \cup Y$ is a sink state. At a state $s \in \widetilde{S}_1$, we have $\widehat T_{X, Y}(s, a, s') = 1$ if and only if $s' = T(s, a)$. At a state $s \in \widetilde{S}_2$, we have $\widehat T_{X, Y}(s, a, s') > 0$ if and only if $a \in \widehat{\srAct}(s)$ and $s' = T(s, a)$. Otherwise, $\widehat T_{X, Y}(s, a, s') = 0$.
		\item $X \cup Y$ is the set of states representing P1's reachability objective.
	\end{itemize}
\end{definition}

It follows by construction that an almost-sure winning strategy of P1 in the hypergame MDP to visit $X \cup Y$ is a stealthy deceptive almost-sure winning strategy.
 
\begin{theorem}
	\label{thm:daswin-reduction}
	P1 can guarantee a visit to $X \cup Y$ from a state $s \in \win_2(\game, F)$ in the true game $\game_{X, Y}^1$ if and only if P1 has an almost-sure winning strategy to visit $X \cup Y$ from the state $s$ in $\widetilde{\hgame}_1(X, Y)$.
\end{theorem}

With this, we can prove the key result of this section: \emph{When players use randomized strategies and the games are analyzed under almost-sure winning condition, fake targets are equally valuable as traps}.

\begin{theorem}
	\label{thm:daswin-equal-for-fakes-traps}
	For any $Z \subseteq \win_2(\game, F) \setminus F$, we have $\daswin_1(Z, \emptyset) = \daswin_1(\emptyset, Z)$.
\end{theorem}
\begin{proof}
 By Lemma~\ref{lma:sract-aswin}, the hypergame \ac{mdp}s $\hgame_1(Z, \emptyset)$ and $\hgame_1(\emptyset, Z)$ are identical. Therefore, P1's almost-sure winning regions in the two hypergames are equal.
\end{proof}

Since the fake targets and traps are equally valuable, Algorithm~\ref{alg:greedymax} can be used to place the decoys in this setting by replacing $\dswin_1(X, Y)$ with $\daswin_1(X, Y)$ on line 6 and 12 of Algorithm~\ref{alg:greedymax}. However, in this case, the complexity of the algorithm is $\calO((V + E)^2 (M+N)^2)$ since the algorithm for computing the almost-sure winning region in the hypergame \ac{mdp} has a time complexity of $\calO((V + E)^2)$ \cite{baier2008principles}.

We conclude this section by establishing that P1 may benefit more from deception when playing against P2 using an almost-sure winning strategy than when playing against P2 using a sure-winning strategy in P2's perceptual game.

\begin{theorem}
	\label{thm:daswin-dswin-comparison}
	For any $X, Y \subseteq \win_2(\game, F) \setminus F$, we have $\daswin_1(X, Y) \subseteq \dswin_1(X, Y)$. 
\end{theorem}
\begin{proof}
	We will establish that, for any state $s \in \daswin_1(X, Y)$, it also belongs to $\dswin_1(X, Y)$. To achieve this, we construct a stealthy deceptive sure-winning strategy $\pi_1^d$ for P1, given any stealthy deceptive almost-sure winning strategy $\pi_1^r$.
	
	Let $\pi_1^d$ be a deterministic strategy such that $\pi_1^d(s) = a$, for some $a \in \supp(\pi_1^r(s))$. 
	
	We will show that $\pi_1^d$ is a stealthy deceptive sure winning strategy for P1. Recall that every stealthy deceptive sure winning strategy is a greedy, deterministic strategy subjectively rationalizable for P2 that ensures a visit to $X \cup Y$ in finitely many steps, regardless of the greedy, deterministic strategy followed by P2. 
	
	\textbf{($\pi_1^d$ is subjectively rationalizable for P2).} $\pi_1^d$ is subjectively rationalizable for P2 whenever $\pi_1^d(s) \in \srAct(s)$. This is indeed the case because the following three conditions hold for all P1 state $s \in \daswin_1(X, Y)$ by definition: (i) $\pi_1^d(s) \in \supp(\pi_1^r(s))$, (ii) $\supp(\pi_1^r(s)) \subseteq \widehat{\srAct(s)}$, and (iii) $\widehat{\srAct(s)} = \srAct(s)$. 
	
	\textbf{($\pi_1^d$ is greedy).} The strategy $\pi_1^d$ is greedy because every action enabled at a P1 state $s \in \daswin_1(X, Y)$ is rank-reducing. This is because every state $s \in \daswin_1(X, Y)$ is also a member of $\win_2(\game, F)$ and Algorithm~\ref{alg:zielonka} includes a P1 state $s$ in $\win_2(\game, F)$, if and only if all actions from $s$ are rank-reducing. 
	
	\textbf{($\pi_1^d$ induces a visit to $X \cup Y$).} We establish that, given any greedy, deterministic P2 strategy $\pi_2^d$, every path $\rho \in \outcomes_{\widehat{\hgame_1}(X, Y)}(s, \pi_1^d, \pi_2^d)$ visits $X \cup Y$ within a finite number of steps. First, we note that  $\outcomes_{\widehat{\hgame}_1(X, Y)}(s, \pi_1^d, \pi_2^d) \subseteq \outcomes_{\widehat{\hgame}_1(X, Y)}(s, \pi_1^r, \pi_2^r)$ holds for any randomized strategy $\pi_2^r$ of P2. This is true because of two facts: (i) $\pi_1^d(s) \in \supp(\pi_1^r(s))$, by definition, and (ii) $\pi_2^d(s) \in \supp(\pi_2^r(s))$, which is true because $\srAct(s) \subseteq \widehat{\srAct(s)}$ holds for all P2 states. Second, we note that, since $\pi_1^r$ is a stealthy deceptive almost-sure winning strategy, every path in $\outcomes_{\widehat{\hgame}_1(X, Y)}(s, \pi_1^r, \pi_2^r)$ eventually visits $X \cup Y$. Clearly, it cannot visit $F$ because all states in $F$ are sink states. Therefore, no path in $\outcomes_{\widehat{\hgame}_1(X, Y)}(s, \pi_1^d, \pi_2^d)$ visits $F$. Since both the strategies $\pi_1^d$ and $\pi_2^d$ are greedy, it follows by Lemma~\ref{lma:effect-on-p1-decoys-in-win2} that $\rho$ must visit $X \cup Y$ within finitely many steps.
\end{proof}

 Figure~\ref{fig:dswin-daswin-comparison} illustrates a toy example where the subset relation is strict, \ie, $\daswin_1(X, Y)$ $\subsetneq \dswin_1(X, Y)$. In this example, $F = \{s_0\}$ is a singleton final state that P2 aims to reach, $X = \{s_1\}$ is the set of traps, and $Y = \{s_2\}$ is the set of fake targets. This results in $\dswin_1(\{s_1\}, \{s_2\}) = \{s_1, s_2, s_4\}$ and $\daswin_1(\{s_1\}, \{s_2\}) = \{s_1, s_2\}$. Notice that $s_4$ is stealthy deceptively sure winning for P1, but not stealthy deceptively almost-sure winning. This is because, when players use greedy deterministic strategies, $b$ is the only action at $s_4$ which is subjectively rationalizable for P2. Since $T(s_4, b) = s_2$ and $s_2$ is a fake target, the game is guaranteed to visit $X \cup Y$. However, when players used randomized strategies, both the actions $b$ and $c$ are subjectively rationalizable for P2 at $s_4$. Thus, the game may reach $s_5$ with a positive probability, from where P1 has no strategy to prevent the game from reaching $F$.


\begin{figure}
	\centering
	\begin{tikzpicture}[->,>=stealth',shorten >=1pt,auto,node distance=2.5cm, scale = 0.65,transform shape]
	
	\node[state,rectangle,accepting] (0) at (0, 0)                         {\huge $s_0$};
	\node[state] (1) at (-3, 2)                          {\huge $s_1$};
	\node[state] (2) at (0, 2)                         {\huge $s_2$};
	\node[state] (3) at (3, 2)                         {\huge $s_3$};
	\node[state,rectangle] (4) at (-3, 4)                         {\huge $s_4$};
	\node[state] (5) at (0, 4)                      {\huge $s_5$};
	\node[state,rectangle] (6) at (3, 4)                         {\huge $s_6$};

		\path 
		(1) edge  node {\huge $a$} (0)
		(2) edge  node {\huge $a$} (0)
		(3) edge  node[above left] {\huge $a$} (0)
		(4) edge  node {\huge $a$} (1)
		(4) edge  node {\huge $b$} (2)
		(4) edge  node {\huge $c$} (5)
		(5) edge  node {\huge $a$} (6)
		(6) edge  node {\huge $a$} (3)
		;
	
\end{tikzpicture}
	\caption{A scenario where $\daswin_1(X, Y) \subsetneq \daswin_1(X, Y)$.} 
	\label{fig:dswin-daswin-comparison}
\end{figure}
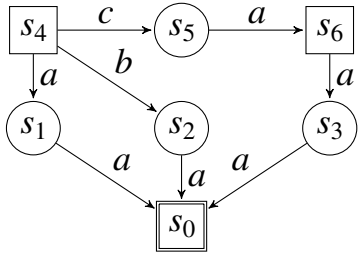


	\section{Experiments}
		\label{sec:experiment}
		We use two experiments to illustrate the key results from our paper. The first experiment employs a gridworld example to demonstrate the proposed Algorithm~\ref{alg:greedymax} and the effectiveness of the decoy placement. The second experiment highlights several key  properties of the decoy placement determined by Algorithm~\ref{alg:greedymax}.

\subsection{Decoy Placement in a Gridworld}


In this experiment, we consider a gridworld example featuring a cat and a mouse as shown in Figure~\ref{fig:tom-jerry-gridworld}. The $7 \times 7$ gridworld has $2$ cheese blocks. The cat is equipped with $M$ mouse traps and $N$ fake cheese blocks to protect the real cheese from the mouse. the mouse's objective is to steal the cheese without getting caught by the cat (the cat captures the mouse when they are simultaneously in the same cell). On the other hand, cat's objective is to place the decoys to safeguard the real cheese strategically. To achieve this, the cat intends to behave in a way that would either lead to the cat capturing the mouse or induce the mouse to visit a decoy. The mouse is assumed to be unaware of the presence of decoys. Both the cat and the mouse can occupy any cell in the gridworld that does not contain an obstacle (black cells). To avoid trivial cases, we assume that the game does not start with the mouse in a cell containing real cheese or a decoy.

A state in the base game between the cat and the mouse is represented as \texttt{(cat.row, cat.col, mouse.row, mouse.col, turn)} that captures the positions (a position is expressed in the row-column format) of the cat and the mouse and the player who selects the next action at that state. At any state, the player whose turn it is to play chooses an action from the set $\{N, E, S, W\}$ and moves to the cell in the intended direction. If the result of the action leads the player to a cell outside the bounds of gridworld or an obstacle, the player returns to the same cell where it started from.

We observe the effect of decoys on the cat's stealthy deceptive sure and almost-sure winning region in the gridworld configuration shown in Figure~\ref{fig:tom-jerry-gridworld} with two blocks of real cheese placed at cells $(1, 6)$ and $(4, 6)$. We consider three scenarios: (A) where $M=2$ and $N=0$, (B) where $M=1$ and $N=1$, and (C) where $M=0$ and $N=2$. This results in the base game's underlying graph having $4050$ states and $16200$ transitions. We use Algorithm~\ref{alg:greedymax} for each scenario to determine the decoy placement under the sure winning criteria. The algorithm solves a total of $85$ hypergames during the two iterations of the \texttt{While} loop (specifically, on lines \texttt{6} and \texttt{13}). The first iteration explores $43$ candidate cells without obstacles or real cheese to determine the placement of the first decoy, while the second iteration explores $42$. The algorithms are implemented in Python 3.10\footnote{The source code is available at \url{https://github.com/abhibp1993/decoy-allocation-problem}.}, and executed on a Windows 10 machine with a core i$7$ CPU running at 3.30GHz and equipped with 32GB of memory.

\begin{figure}[tb]
	\centering
	\includegraphics[scale=0.30]{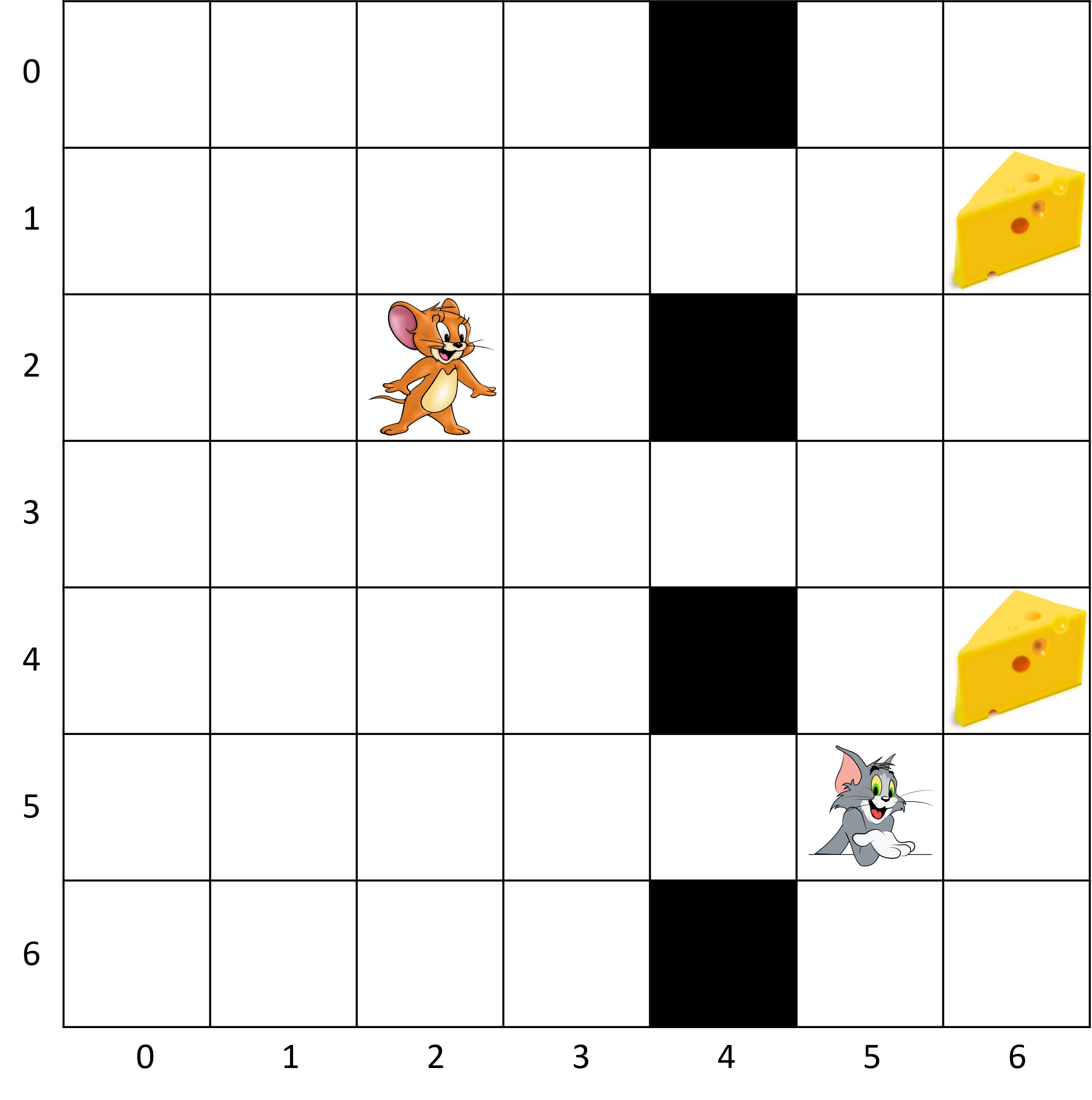}
	\caption{Gridworld example with a cat and a mouse with $2$ cheese blocks.}
	\label{fig:tom-jerry-gridworld}
\end{figure}

To measure and compare the effectiveness of a given placement of traps and fake targets during the iterations of Algorithm~\ref{alg:greedymax}, we introduce a real-valued metric called \emph{value of deception}. Intuitively, the value of deception measures the proportion of P2's winning states in the base game $\game$ that become winning for P1 in the hypergame $\widehat{\hgame}(X, Y)$ or $\widetilde{\hgame}(X, Y)$. Under the stealthy deceptive sure winning condition, when $\win_2(\game, F) \neq F$, the value of deception is defined as follows:
\begin{align*}
	\vod(X, Y) = \frac{|\dswin_1(X, Y)|}{|\win_2(\game, F)| - |F|}
\end{align*}
If $\win_2(\game, F) = F$, \ie, when no states apart from the final states are winning for P2 in $\game$, we set $\vod(X, Y) = 0$. The value of deception is defined analogously when the interaction is analyzed under an almost-sure winning criterion.

We analyze the key insights obtained by solving $85$ hypergames and examining the resulting value of deception. Figure~\ref{fig:results} depicts a heatmap, where the value displayed in each cell denotes the value of deception achieved by allocating the next decoy in that cell. The value in each cell is computed based on the map $Z$ constructed during each of the two iterations of Algorithm~\ref{alg:greedymax}. The figure includes two heatmaps each for the three scenarios (A), (B), and (C). Specifically, Figures~\ref{fig:sub-a} and \ref{fig:sub-b} depict the heatmaps corresponding to the first and second iteration of the algorithm for scenario (A). Similarly, Figures~\ref{fig:sub-c} and \ref{fig:sub-d} show the two heatmaps for scenario (B), and Figures~\ref{fig:sub-e} and \ref{fig:sub-f} for scenario (C).


In Figure~\ref{fig:sub-a}, the cell values indicate the value of deception achieved by placing the first trap at each respective cell. For instance, the value $0.28$ in cell $(1, 5)$ indicates the value of deception obtained by placing the first trap at that location. The first trap is positioned at $(1, 5)$ as it is the highest value. In Figure~\ref{fig:sub-b}, the cell values indicate the combined value of deception achieved by placing the second trap at a given cell in addition to the trap selected in the first iteration. For instance, the value $0.5$ in cell $(5, 5)$ represents the value of deception obtained by placing two traps: the first trap at location $(1, 5)$ (as determined in the first iteration) and the second trap at $(5, 5)$. The second trap is placed there since the maximum deception value is observed at $(5, 5)$. The heatmaps in Figures~\ref{fig:sub-c}-\ref{fig:sub-f} are understood in a similar manner.

We now discuss key observations and insights from Figure~\ref{fig:results}. First, observe that when only traps are placed (Figures~\ref{fig:sub-a}, \ref{fig:sub-b}, and \ref{fig:sub-d}), the value of deception increases as we move closer to the real cheese. This is because traps cut the mouse's winning paths to real cheese. For instance, in Figure~\ref{fig:sub-a}, suppose that the mouse starts from a cell in row $1$ and the cat starts from a cell $(4, 1)$. Then, the mouse has a sure winning strategy to steal the cheese at $(1, 6)$. Now, consider two placements of the first trap: $(1, 1)$ and $(1, 5)$. The trap at $(1, 1)$ will be effective only if the mouse starts at $(1, 0)$ since if the mouse begins from a cell to the right of $(1, 1)$, she is guaranteed to visit $(1, 6)$ without being trapped or caught. On the other hand, the trap at $(1, 5)$ will be effective whenever the mouse starts between $(1, 0)$ and $(1, 4)$ because every path induced by any of her sure winning strategies to visit $(1, 6)$ from these initial positions passes through $(1, 5)$. Hence, placing a trap at $(1, 5)$ yields a higher value of deception than placing it at $(1, 1)$.

In contrast, fake cheese attracts the mouse by providing an alternative to visiting the real cheese. Therefore, when placing the fake cheese, the value of deception increases as we move closer to the fake cheese. For instance, in Figure~\ref{fig:sub-c}, we notice that the values in cells $(2, 3)$ and $(3, 3)$ are higher than their neighboring cells. This is because, when fake cheese is present at either of these cells, the mouse believes there are three cheese blocks in the game instead of two. Consequently, when the cat starts at $(5, 1)$ and the mouse starts at any cell with row coordinates of $0, 1, 2$ and column coordinates of $0, 1, 2$, the mouse's subjectively rationalizable sure winning strategy would lead him to visit either the fake cheese at $(2, 3)$ or $(3, 3)$ instead of the real cheese at $(1, 6)$ or $(4, 6)$. Since the highest value of deception is observed at cell $(3, 3)$, the cat places the first fake cheese in that cell.

The results also confirm our conclusion that \emph{fake targets have a higher value than traps when the game is analyzed under sure winning condition}. To see this, compare the value of deception for any cell in Figure~\ref{fig:sub-d} and Figure~\ref{fig:sub-f}, and Figure~\ref{fig:sub-d} and Figure~\ref{fig:sub-f}. We observe that the value in the second heatmap (where a fake cheese is placed in the cell) is greater than or equal to that in the first heatmap (where a trap is placed in the cell).

\subsection{Effectiveness of Decoy Placement under Sure and Almost-sure Winning Conditions}

In this second experiment, we compare the effectiveness of placing traps versus fake targets under stealthy, deceptive sure and almost-sure winning conditions. We employ randomly generated graphs to explore interesting case studies. Each game consists of $150$ states, of which $75$ are P1 states, and the remaining are P2 states. At every state in each game, we randomly select an integer between $1$ and $5$ to determine the number of actions enabled at that state. Subsequently, the next state on performing each enabled action at a given state is determined at random. 

With these exploratory experiments, we focus our analysis on four games on graphs as these present interesting results.
For each of the four games, we use Algorithm~\ref{alg:greedymax} to determine decoy placement and compute the corresponding value of deception under four conditions: (i) placing $5$ traps under stealthy deceptive sure winning condition, (ii) placing $5$ fake targets under stealthy deceptive sure winning condition, (iii) placing $5$ traps under stealthy deceptive almost-sure winning condition, and (iv) placing $5$ fake targets under stealthy deceptive almost-sure winning condition. Figure~\ref{fig:exp2-results} depicts the variation in the value of deception for cases (i)-(iv) as we progressively introduce the traps or fake targets in four selected games. 

\begin{figure*}[!th]
	\begin{multicols}{2}
		\centering
		\begin{subfigure}{\linewidth}
			\includegraphics[width=\linewidth]{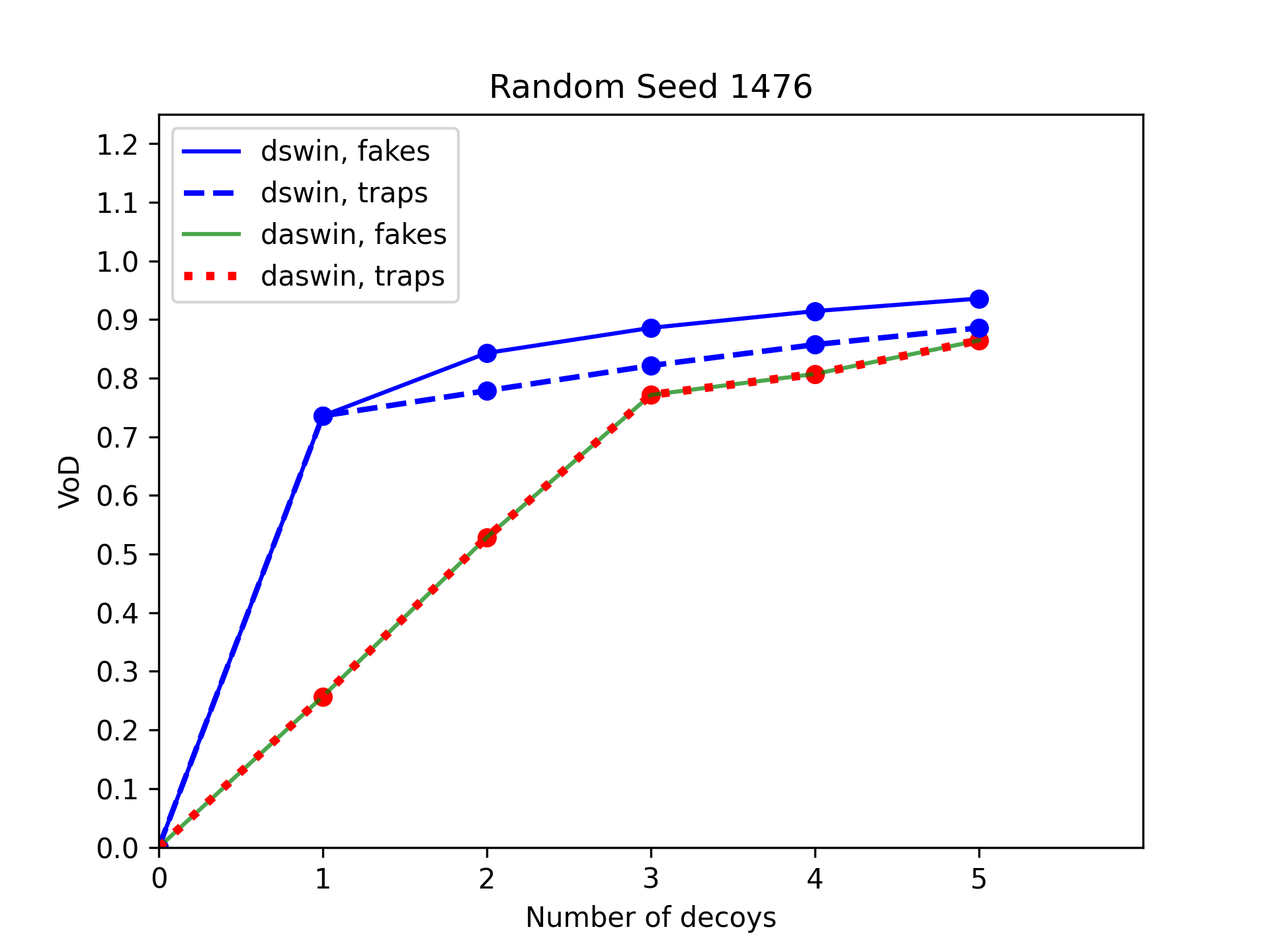}
			\caption{} 
			\label{fig:exp3-sub-a}
		\end{subfigure}
		\par
		\begin{subfigure}{\linewidth}
			\includegraphics[width=\linewidth]{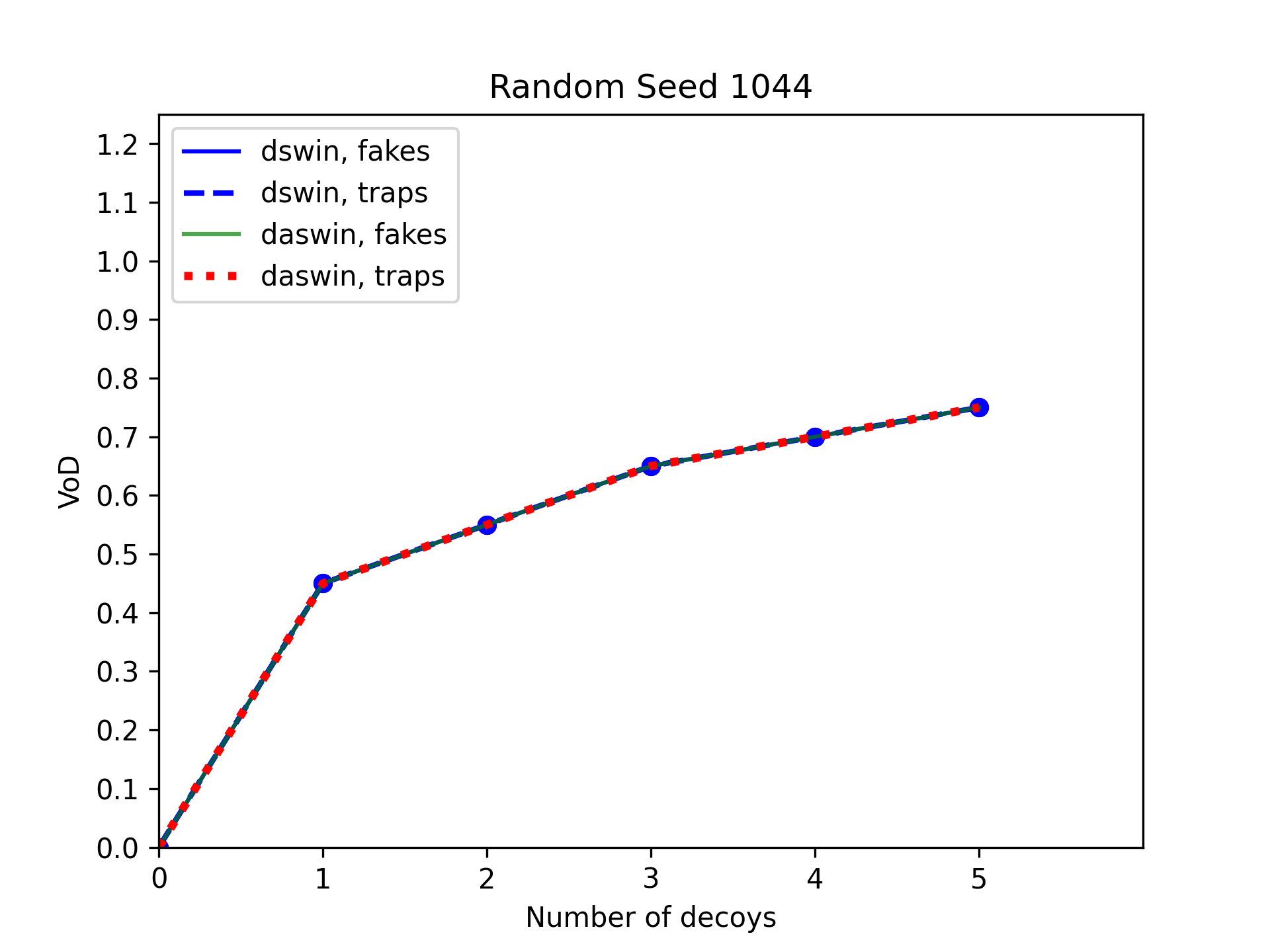}
			\caption{} 
			\label{fig:exp3-sub-b}
		\end{subfigure}
		\par
		\begin{subfigure}{\linewidth}
			\includegraphics[width=\linewidth]{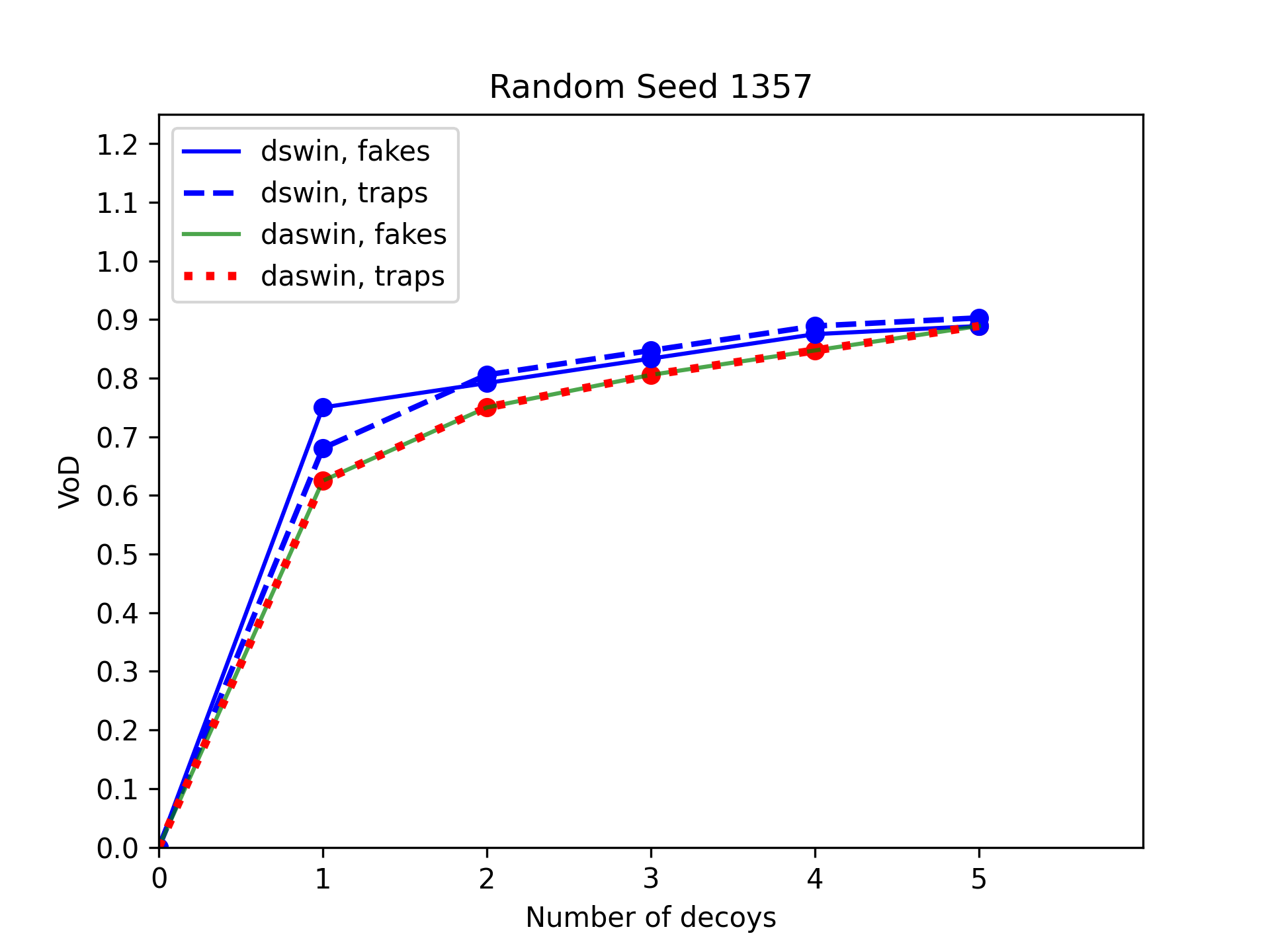}
			\caption{} 
			\label{fig:exp3-sub-c}
		\end{subfigure}
		\par
		\begin{subfigure}{\linewidth}
			\includegraphics[width=\linewidth]{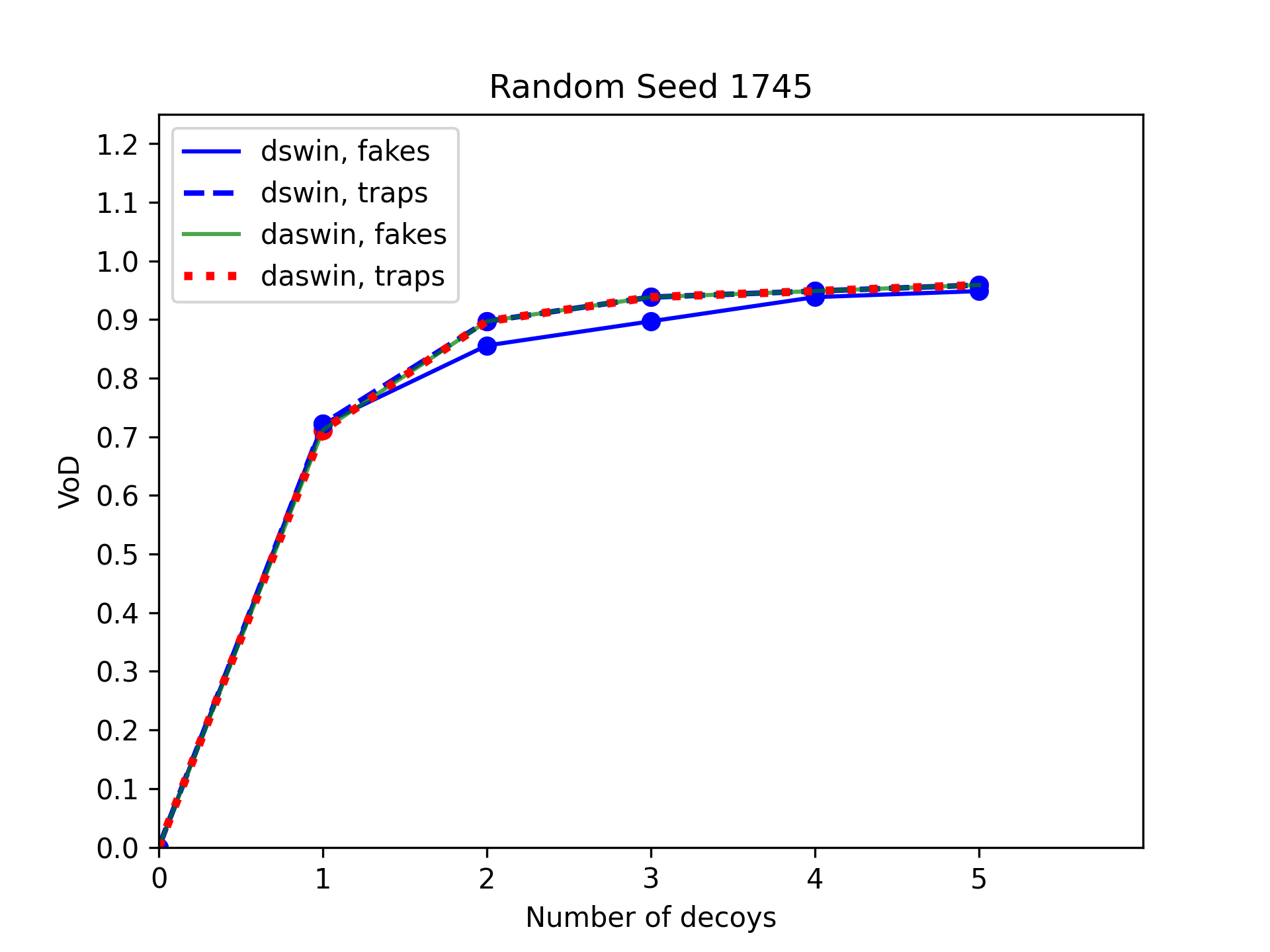}
			\caption{} 
			\label{fig:exp3-sub-d}
		\end{subfigure}
	\end{multicols}
	\caption{The value of deception obtained by placing traps and fake targets under stealthy deceptive sure and almost-sure winning conditions in four selected games.}
	\label{fig:exp2-results}
\end{figure*}

Figures~\ref{fig:exp3-sub-a} and \ref{fig:exp3-sub-b} present instances that align with our theoretical findings. Since the dashed blue line remains at par or below the solid blue line, we observe that the value of deception obtained by placing fake targets is greater than or equal to that obtained by placing traps, both under stealthy deceptive sure winning condition. This confirms the findings in Theorem~\ref{thm:traps-fakes-subset-relation}. Furthermore, the overlapping of the red dotted line and green lines indicates that placing traps and fake targets under stealthy deceptive almost-sure winning condition yield the same value of deception, which is aligned with the findings of Theorem~\ref{thm:daswin-equal-for-fakes-traps}. Lastly, the outcomes also align with the implications outlined in Theorem~\ref{thm:daswin-dswin-comparison}, as both the red-dotted and green lines consistently remain positioned below the blue lines. Consequently, the value of deception achieved under the sure winning condition is consistently greater than or equal to that attained under the almost-sure winning condition. Figure~\ref{fig:exp3-sub-b} presents a special case wherein the intrinsic topology of the game graph leads to a convergence of deception values across all four cases (i)-(iv).

Figures~\ref{fig:exp3-sub-c} and \ref{fig:exp3-sub-d} present instances where the results appear to diverge from our theoretical predictions. In Figure~\ref{fig:exp3-sub-c}, we encounter a situation where the value of deception achieved under the sure winning condition by strategically placing traps is greater than the value obtained by placing fake targets. This outcome seemingly contradicts the assertions made in Theorem~\ref{thm:traps-fakes-subset-relation}. In Figure~\ref{fig:exp3-sub-d}, we encounter another scenario where the value of deception obtained by deploying either traps or fake targets under the almost-sure winning condition exceeds the value attained by placing fake targets under the sure winning condition, thereby deviating from the anticipated results stipulated in Theorem~\ref{thm:daswin-dswin-comparison}. However, these disparities can be attributed to the greedy approach employed by Algorithm~\ref{alg:greedymax}. For instance, in Figure~\ref{fig:exp3-sub-c}, Algorithm~\ref{alg:greedymax} determined the states \texttt{s22, s80} as the first two fake targets and \texttt{s101, s74} as the first two traps. To understand these choices, let us examine the values of deception for the following placements:
\begin{alignat*}{2}
		&\vod(\emptyset, \{\mathtt{s22}\}) = 0.7500, \qquad &&\vod(\emptyset, \{\mathtt{s101}\}) = 0.6805 \\ 
		&\vod(\{\mathtt{s22}\}, \emptyset) = 0.4166, \qquad &&\vod(\{\mathtt{s101}\}, \emptyset) = 0.6805 \\ 
		&\vod(\emptyset, \{\mathtt{s101}, \mathtt{s74}\}) = 0.8055, \qquad &&\vod(\emptyset, \{\mathtt{s22}, \mathtt{s80}\}) = 0.7916
\end{alignat*}
We observe that the value of deception attained by placing fake targets at \texttt{s101, s74} is higher than that obtained by placing them at \texttt{s22, s80}. Thus, we would expect the algorithm to select the latter states to be the fake targets. However, the Algorithm~\ref{alg:greedymax} follows a greedy approach.  Since the value of deception when the first fake target is placed at \texttt{s22} is greater than when it is placed at all other states, including \texttt{s101}, \texttt{s22} is selected as the first fake target. Given the first fake target, the choice of the second fake target that yields that maximum value of deception is \texttt{s80}. In other words, the deviation from theoretical expectations is due to the sub-optimal placement suggested by the greedy algorithm.

We conclude by noting that the value of deception increases monotonically until the value of $1.0$ is attained. In any game, the value of $1.0$ is guaranteed to be achieved if there is no bound on the number of decoys. In the worst case (for example, consider star topology), a decoy must be placed at every state for the value of deception to be one. 


\begin{figure*}[!ht]
	\centering
	\begin{multicols}{2}
		\centering
		\begin{subfigure}{\linewidth}
			\includegraphics[width=0.9\linewidth]{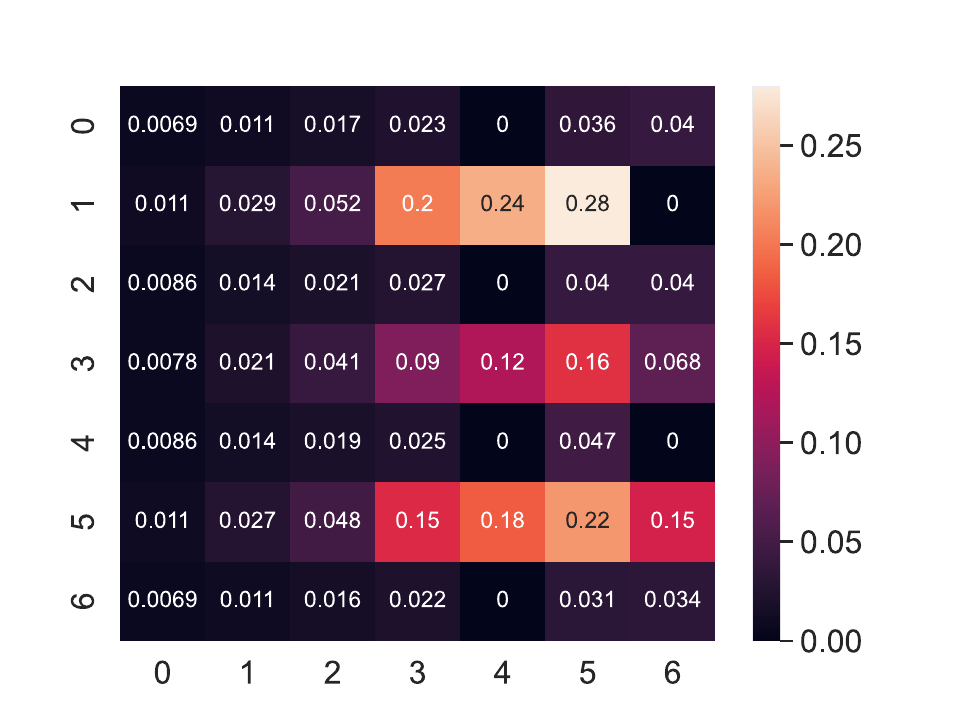}
			\caption{Scenario (A): Value of deception when the first trap is placed within the given cell.}
			\label{fig:sub-a}
		\end{subfigure}
		\par
		\begin{subfigure}{\linewidth}
			\includegraphics[width=0.9\linewidth]{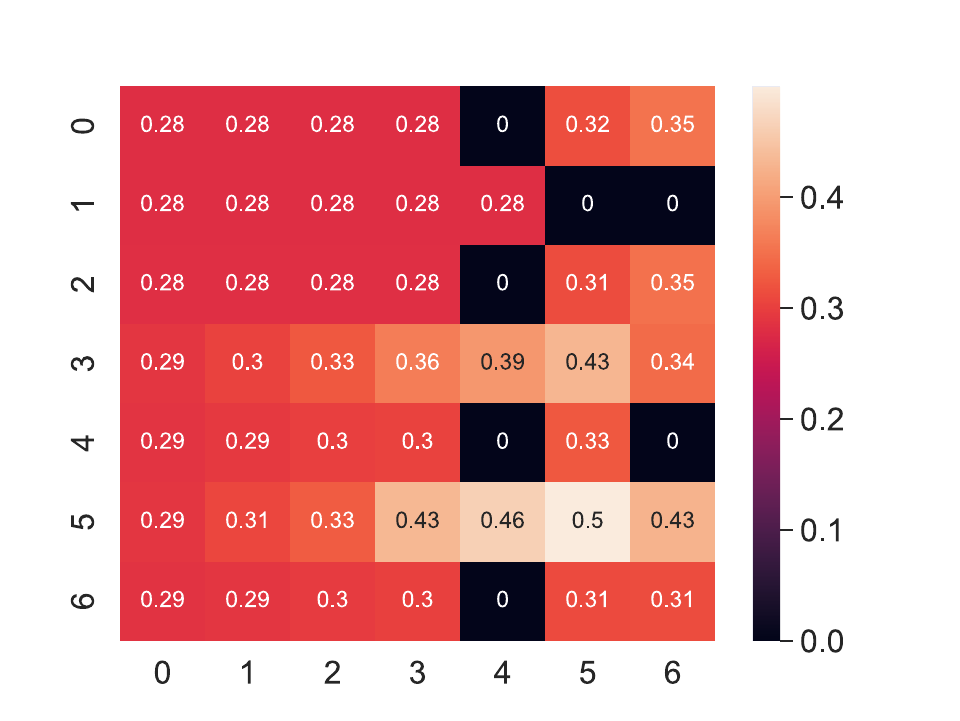}
			\caption{Scenario (A): Value of deception when first trap is placed at $(1, 5)$ and second trap is placed within the given cell.}
			\label{fig:sub-b}
		\end{subfigure}
		\par
		\begin{subfigure}{\linewidth}
			\includegraphics[width=0.9\linewidth]{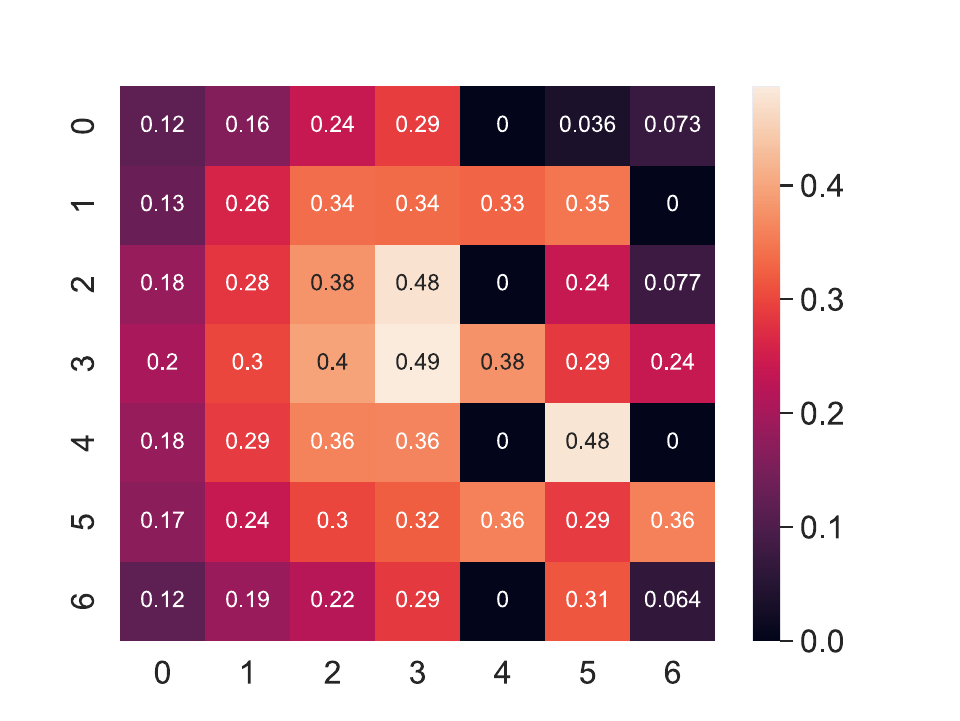}
			\caption{Scenario (B): Value of deception when first fake cheese is placed within the given cell.}
			\label{fig:sub-c}
		\end{subfigure}
		\par
		\begin{subfigure}{\linewidth}
			\includegraphics[width=0.9\linewidth]{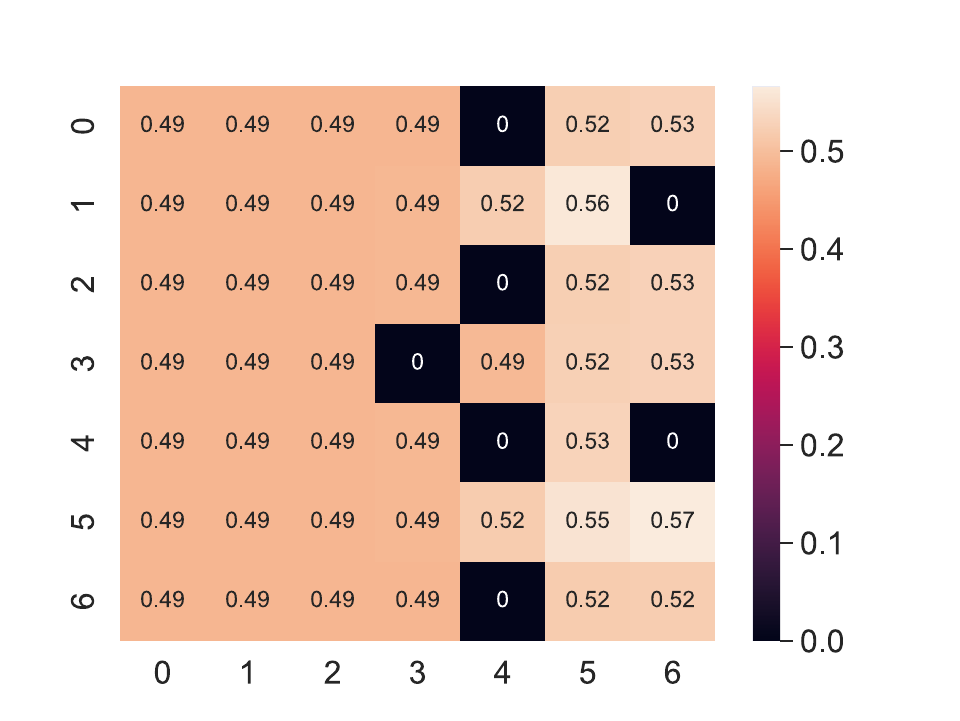}
			\caption{Scenario (B): Value of deception when first fake cheese is placed at \replaced{$(3, 3)$}{$(4, 5)$} and a trap is placed within the given cell.}
			\label{fig:sub-d}
		\end{subfigure}
		\par
		\begin{subfigure}{\linewidth}
			\includegraphics[width=\linewidth]{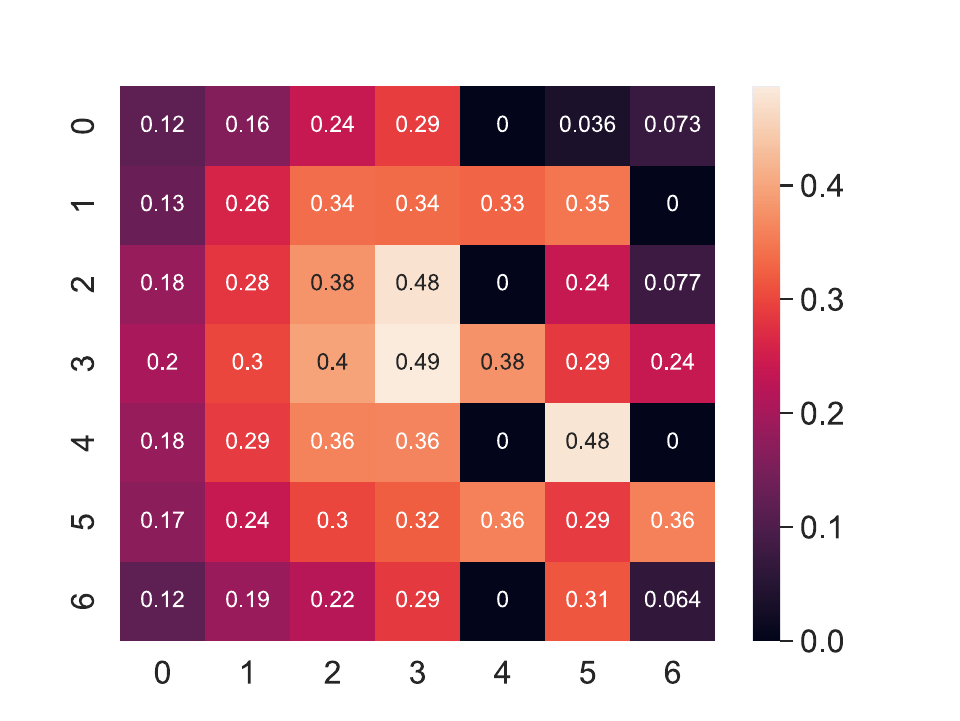}
			\caption{Scenario (C): Value of deception when first fake cheese is placed within the given cell.}
			\label{fig:sub-e}
		\end{subfigure}
		\par
		\begin{subfigure}{\linewidth}
			\includegraphics[width=0.9\linewidth]{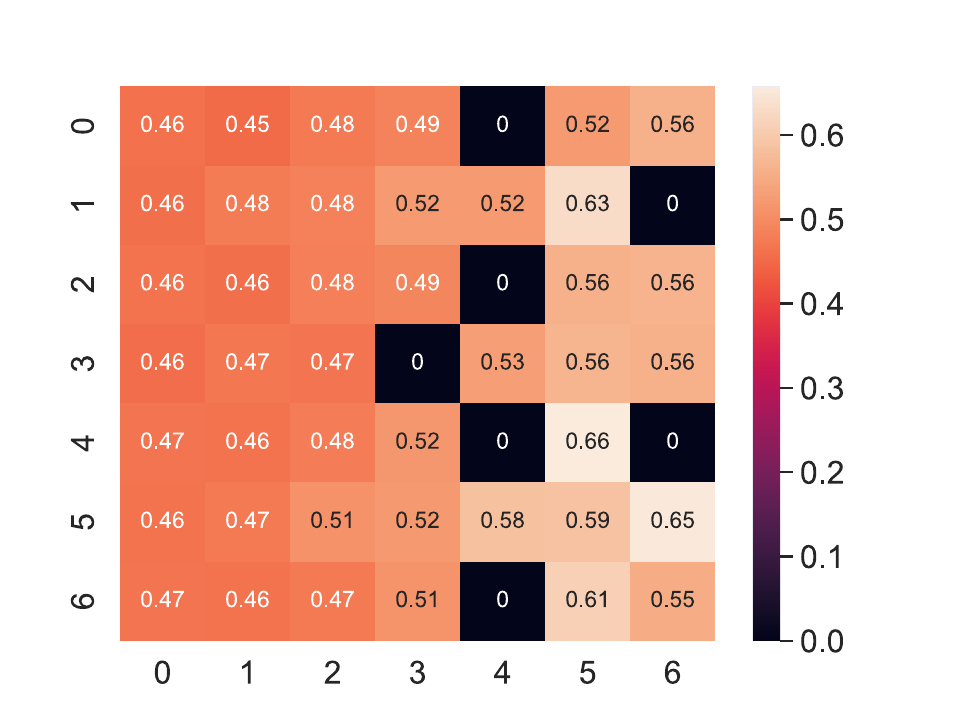}
			\caption{Scenario (C): Value of deception when first fake cheese is placed at $(4, 5)$ and the second fake cheese is placed within the given cell.}
			\label{fig:sub-f}
		\end{subfigure}
	\end{multicols}
	\caption{The values of deception compared by Algorithm~\ref{alg:greedymax} in each of the two iterations to determine the two decoys for scenarios (A)-(C).}
	\label{fig:results}
\end{figure*}

	\section{Conclusion}
	\label{sec:conclusion}
	
	We have studied the problem of optimally allocating two types of deception resources (decoys) and jointly synthesizing a deceptive strategy for P1 in a game on graph with incomplete information. The resulting decoy placement disinforms P2 about the game dynamics, and the deceptive strategy leverages P2's lack of information to enforce a visit to a decoy state. Our contribution addresses two significant challenges: First, we introduce a hypergame on graph model to capture P2's incomplete information and P1's awareness of P2's misperception. We define two solution concepts: stealthy deceptive sure-winning and stealthy deceptive almost-sure-winning strategies that help identify the P1 strategies subjectively rationalizable for P2 and allow P1 to enforce a visit to a decoy state. These solution concepts establish a connection between the solution concept of subjective rationalizability, which is primarily studied for normal-form games with incomplete information, and sure/almost-sure winning, which is defined for the qualitative analysis of games on graphs, also referred to as $\omega$-regular games. Notably, these strategies are designed to be stealthy, as they only employ actions that are subjectively rationalizable for P2 in its perceptual game. This ensures that P2 remains unaware of the deception being employed. Second, in order to efficiently search for an optimal decoy placement in a combinatorially large space, we used compositional synthesis from formal methods and proved that the objective function of the optimal decoy allocation problem is a monotone, non-decreasing, and under certain conditions, sub- or super-modular. Using this fact, we proposed a greedy algorithm, which runs in polynomial time, to compute a decoy placement, which is $(1- 1/e)$-approximate when the objective function is either sub- or super-modular.
	
	We identify two directions for future directions. One direction is to drop the \emph{stealthiness} requirement for P1's strategy. In practice, there are several applications, \eg, conflict \cite{fraser1979solving}, where stealthiness is essential. However, there also exist applications, \eg, human-robot interaction \cite{shim2013taxonomy}, where stealthiness is not necessary. In such cases, the model should capture P2's ability to learn about the decoys during the interaction. Another direction is to consider a similar placement problem for the class of concurrent stochastic games on graphs \cite{deAlfaro2007concurrent}. 

\section*{Declaration of Generative AI and AI-assisted technologies in the writing process}
	During the preparation of this work, the author(s) used ChatGPT in order to improve the readability and language. After using this tool/service, the author(s) reviewed and edited the content as needed and take(s) full responsibility for the content of the publication.
	
\section*{Acknowledgements}
Research was sponsored by the Army Research Office and was accomplished under Grant Number W911NF-22-1-0166 and W911NF-22-1-0034 and in part by NSF under grant No. 2144113. The views and conclusions contained in this document are those of the authors and should not be interpreted as representing the official policies, either expressed or implied, of the Army Research Office or the U.S. Government. The U.S. Government is authorized to reproduce and distribute reprints for
Government purposes notwithstanding any copyright notation herein.

	\bibliographystyle{elsarticle-num-names} 
	\bibliography{main_aij}

\end{document}